\documentclass{sig-alternate-05-2015}

\newcommand{\rronly}[1]{#1}
\newcommand{\pponly}[1]{}

\newcommand{\dpcmt}[1]{{\color{magenta} [{#1}]}}
\newcommand{\pscmt}[1]{{\color{blue} [{#1}]}}

\usepackage{etex}

\usepackage{amssymb,stmaryrd}
\usepackage{paralist}
\usepackage{hyperref}
\usepackage{galois}
\usepackage{verbatim}
\usepackage{mathrsfs}
\usepackage{graphicx}
\usepackage[figurename=Fig.]{caption}
\usepackage{subcaption}
\usepackage{booktabs}
\usepackage{tabularx}
\usepackage{changepage}

\let\sigproof\proof\let\proof\relax
\let\sigendproof\endproof\let\endproof\relax

\usepackage{amsthm}

\let\proof\sigproof
\let\endproof\sigendproof

\theoremstyle{definition}

\newtheorem{theorem}{Theorem}[section]

\newtheorem{definition}{Definition}

\newtheorem{lemma}[theorem]{Lemma}

\newtheorem{property}{Property}

\usepackage{comment}
\pponly{\excludecomment{theappendix}}
\rronly{\includecomment{theappendix}}

\usepackage{fancybox}

\usepackage{tikz}
\usetikzlibrary{calc}
\usetikzlibrary{shapes}
\usetikzlibrary{shapes.multipart}
\usetikzlibrary{shadows}
\usetikzlibrary{arrows}
\usetikzlibrary{patterns}
\usetikzlibrary{chains}

\usepackage{pgfplots}

\usepackage[ruled,linesnumbered,vlined]{algorithm2e}
\SetKwFor{Func}{function}{}{}
\SetKw{Continue}{continue}
\DontPrintSemicolon
\SetKwInOut{Input}{Input}

\usepackage{listings}
\makeatletter
\lst@Key{countblanklines}{true}[t]%
    {\lstKV@SetIf{#1}\lst@ifcountblanklines}

\lst@AddToHook{OnEmptyLine}{%
    \lst@ifnumberblanklines\else%
       \lst@ifcountblanklines\else%
         \advance\c@lstnumber-\@ne\relax%
       \fi%
    \fi}
\makeatother

\makeatletter
\newcommand*\Suppressnumber{%
  \lst@AddToHook{OnNewLine}{%
    \let\thelstnumber\relax%
     \advance\c@lstnumber-\@ne\relax%
    }%
}
\newcommand*\Reactivatenumber{%
  \lst@AddToHook{OnNewLine}{%
   \let\thelstnumber\origthelstnumber%
   \advance\c@lstnumber\@ne\relax}%
}
\makeatother

\newcommand{\Omit}[1]{}

\newcommand{\true}{\mathit{true}}
\newcommand{\false}{\mathit{false}}

\newcommand{\prog}{\mathit{Prog}}
\newcommand{\vars}{\mathit{Vars}}
\newcommand{\objs}{\mathit{Objs}}
\newcommand{\func}{f}
\newcommand{\funcs}{\mathit{Funcs}}
\newcommand{\threads}{\mathit{Threads}}
\newcommand{\nodes}{\mathit{Nodes}}
\newcommand{\lock}{\mathit{lock}}
\newcommand{\unknownlock}{\star}
\newcommand{\locks}{\mathit{Locks}}
\newcommand{\currnode}{\mathit{CurrentNode}}
\newcommand{\sign}{\mathit{Sign}}
\newcommand{\labels}{\mathit{Labels}}

\newcommand{\loc}{\ell}
\newcommand{\locs}{\mathit{L}}
\newcommand{\edges}{\mathit{E}}

\newcommand{\fp}{\mathit{fp}}
\newcommand{\fpm}{\mathit{fpm}}
\newcommand{\fpms}{\mathit{Fpms}}
\newcommand{\param}{\mathit{par}}

\renewcommand{\arg}{\mathit{arg}}

\newcommand{\place}{p}
\newcommand{\places}{P}

\newcommand{\src}{\textsf{src}}
\newcommand{\tgt}{\textsf{tgt}}

\newcommand{\customanalysis}{a}
\newcommand{\fsframework}{s}
\newcommand{\fiframework}{i}

\newcommand{\statea}{s_{\customanalysis}}
\newcommand{\states}{s_{\fsframework}}
\newcommand{\statei}{s_{\fiframework}}
\newcommand{\doma}{\mathcal{D}_{\customanalysis}}
\newcommand{\doms}{\mathcal{D}_{\fsframework}}
\newcommand{\domi}{\mathcal{D}_{\fiframework}}
\newcommand{\transff}{\mathcal{T}}
\newcommand{\transfaf}{\mathcal{T}_{\customanalysis}}
\newcommand{\transfsf}{\mathcal{T}_{\fsframework}}
\newcommand{\transfif}{\mathcal{T}_{\fiframework}}
\newcommand{\transf}[1]{\transff\llbracket #1 \rrbracket}
\newcommand{\transfa}[1]{\transfaf\llbracket #1 \rrbracket}
\newcommand{\transfs}[1]{\transfsf\llbracket #1 \rrbracket}
\newcommand{\transfi}[1]{\transfif\llbracket #1 \rrbracket}
\newcommand{\joina}{\join_{\customanalysis}}
\newcommand{\joins}{\join_{\fsframework}}
\newcommand{\joini}{\join_{\fiframework}}
\newcommand{\bota}{\bot_{\customanalysis}}

\newcommand{\pop}{\textsf{pop}}

\newcommand{\commonprefix}{\textsf{common\_prefix}}
\newcommand{\haspath}{\textsf{has\_path}}
\newcommand{\inloop}{\textsf{in\_loop}}
\newcommand{\onallpaths}{\textsf{on\_all\_paths}}
\newcommand{\match}{\textsf{match}}

\newcommand{\placetop}{\textsf{top}}
\newcommand{\entryloc}{\textsf{entry\_loc}}
\newcommand{\exitloc}{\textsf{exit\_loc}}

\newcommand{\entry}{\textsf{entry}}

\newcommand{\placenext}{\textsf{next}}
\newcommand{\placenexts}{\textsf{next}_{\fsframework}}
\newcommand{\placenexti}{\textsf{next}_{\fiframework}}

\newcommand{\maylockset}{\mathit{ls}_a}
\newcommand{\mustlockset}{\mathit{ls}_u}
\newcommand{\lockset}{\mathit{ls}}

\newcommand{\symbols}{\textsf{symbols}}

\newcommand{\cfaedges}{E}
\newcommand{\cfaedge}{e}

\newcommand{\getfunc}{\textsf{func}}
\newcommand{\getthread}{\textsf{get\_thread}}

\newcommand{\isfunc}{\textsf{is\_func}}
\newcommand{\isfp}{\textsf{is\_func\_pointer}}

\newcommand{\op}{\textsf{op}}
\newcommand{\lockop}{\textsf{lock}}
\newcommand{\unlockop}{\textsf{unlock}}
\newcommand{\createop}{\textsf{create}}
\newcommand{\joinop}{\textsf{join}}

\newcommand{\funcentry}{\textsf{func\_entry}}
\newcommand{\funcexit}{\textsf{func\_exit}}
\newcommand{\threadexit}{\textsf{thread\_exit}}
\newcommand{\threadjoin}{\textsf{thread\_join}}
\newcommand{\threadentry}{\textsf{thread\_entry}}

\newcommand{\matchfp}{\textsf{match\_fp}}

\newcommand{\dom}{\textsf{dom}}

\newcommand{\join}{\sqcup}

\newcommand{\joinfp}{\join_{\mathit{fp}}}

\newcommand*{\nolink}[1]{%
  \begin{NoHyper}#1\end{NoHyper}%
}

\begin{document}

\toappear{}

\setcopyright{acmcopyright}

\doi{}

\isbn{}

\conferenceinfo{ASE'16}{ASE'16}

\acmPrice{}

\conferenceinfo{ASE'16}{ASE'16}

\title{Sound Static Deadlock Analysis for C/Pthreads \rronly{\\(Extended Version)}
}

\numberofauthors{4}

\author{
\alignauthor Daniel Kroening\\
       \affaddr{University of Oxford}\\
       \affaddr{Oxford, UK}\\
       \email{kroening@cs.ox.ac.uk}
\alignauthor Daniel Poetzl\\
       \affaddr{University of Oxford}\\
       \affaddr{Oxford, UK}\\
       \email{daniel.poetzl@cs.ox.ac.uk}
\alignauthor Peter Schrammel\\
       \affaddr{University of Sussex}\\
       \affaddr{Brighton, UK}\\
       \email{p.schrammel@sussex.ac.uk}
\and
\alignauthor Bj\"orn Wachter\\
       \affaddr{University of Oxford}\\
       \affaddr{Oxford, UK}\\
       \email{bjoern.wachter@gmail.com}
}

\maketitle


\begin{abstract}
We present a static deadlock analysis approach for C/pthreads. The design of our
method has been guided by the requirement to analyse real-world code. Our
approach is sound (i.e., misses no deadlocks) for programs that have defined
behaviour according to the C standard, and precise enough to prove
deadlock-freedom for a large number of programs. The method consists of a
pipeline of several analyses that build on a new context- and thread-sensitive
abstract interpretation framework. We further present a lightweight dependency
analysis to identify statements relevant to deadlock analysis and thus speed up
the overall analysis. In our experimental evaluation, we succeeded to prove
deadlock-freedom for 262 programs from the Debian GNU/Linux distribution with in
total 2.6\,MLOC in less than 11 hours.
\end{abstract}

%
%

\ccsdesc[500]{Theory of computation~Program analysis}
%
%

%
%
\printccsdesc


\keywords{deadlock detection, lock analysis, static analysis, abstract interpretation}


\section{Introduction}

Locks are the most frequently used synchronisation mechanism in concurrent
programs to guarantee atomicity, prevent undefined beha\-vi\-our, and hide
weak-memory effects of the
underlying architectures.  However, locks, if not correctly used, can cause
a \emph{deadlock}, where one thread holds a lock that the other one needs
and vice versa.  Deadlocks can have disastrous consequences such as the
Northeastern Blackout~\cite{blackout},
where a monitoring system froze, and prevented detection of a local power
problem, which brought down the power network of the Northeastern
United States.

In small programs, deadlocks may be spotted easily.  However, this is not
the case in larger software systems.  So called \emph{locking disciplines} aim at
preventing deadlocks but are difficult to maintain as the system evolves, and
every extension bears the risk of introducing deadlocks.  For example, if
the order in which locks are acquired is different in different parts of the
code this may cause a deadlock.  
%
%

The problem is exacerbated by the fact that deadlocks are difficult to
discover by means of testing.  Even a test suite with full line coverage is
insufficient to detect all deadlocks, and similar to other concurrency bugs,
triggering a deadlock requires a specific thread schedule and set of
particular program inputs.  Therefore, static analysis is a promising
candidate for a thorough check for deadlocks.  The challenge is to devise a
method that is scalable and yet precise enough. The inherent trade-off
between scalability and precision has gated static approaches to deadlock
analysis, and many previously published results had to resort to unsound
approximations in order to obtain scalability (we provide a
\rronly{comprehensive }survey of the existing static approaches in
Sec.~\ref{sec:related}).

We hypothesise that static deadlock detection can be performed with a
sufficient degree of precision and scalability and without sacrificing
soundness.  To this end, this paper presents a new method for statically
detecting deadlocks.  To~quantify scalability, we have applied our
implementation to a large body of real-world concurrent code from the Debian
GNU/Linux project.


\medskip\noindent
Specifically, this paper makes the following contributions:

\medskip
\begin{compactenum}
\item The first, to the best of our knowledge, \emph{sound} static deadlock
analysis approach for C/pthreads that can handle real-world code.


\item A new context- and thread-sensitive abstract interpretation framework that
forms the basis for the analyses that comprise our approach. The framework unifies
contexts, threads, and program locations via the concept of a place.


\item A novel lightweight dependency analysis which identifies statements that
could affect a given set of program expressions. We use it to speed up the
pointer analysis by focusing it to statements that are relevant to deadlock
analysis.


\item We show how to build a lock graph that soundly captures a variety of
sources of imprecision, such as may-point-to information and thread
creation in loops/recursions, and how to combine the cycle detection
with a non-concurrency check to prune infeasible cycles.



\item A thorough experimental evaluation on 715 programs from
Debian GNU/Linux with 7.1\,MLOC in total and up to 50KLOC per program. 
%
\end{compactenum}

\section{Overview}\label{sec:overview}

\begin{figure*}
\begin{minipage}{0.34\textwidth}
\scriptsize
\begin{lstlisting}
void vlc_mutex_lock (vlc_mutex_t *p) {
  int val = pthread_mutex_lock(p);
  VLC_THREAD_ASSERT("locking mutex");
}
\end{lstlisting}
\vspace{-2.5ex}
\begin{lstlisting}
void create_worker(void *(*func)(void *),
                          void *arg) {
    pthread_attr_t  attr;
    int ret;
    pthread_attr_init(&attr);
    if ((ret=pthread_create(
               &((THREAD*)arg)->thread_id,
               &attr, func, arg)) != 0) {
        fprintf(stderr, "Error: %s\n",
                strerror(ret));
        exit(1);
    }
}
\end{lstlisting}
\vspace{-4ex}
\caption{Lock and create wrappers}
\label{fig:wrappers}
\end{minipage}%
\begin{minipage}{0.2\textwidth}
\scriptsize
\centering
\begin{tikzpicture}
\tikzset{
nodestyle/.style={draw, rounded corners, minimum width=1cm, text width=2.05cm, minimum height=0.75cm, align=center, inner sep=1mm},
edgestyle/.style={->, font=\small}
}
\node[nodestyle] (a) at (0, 0) {function pointer call removal};
\node[nodestyle] (b) at (0, -1.25) {return removal};
\node[nodestyle] (c) at (0, -2.5) {ICFA construction};

\node (d) at (0, 1) {C program};
\node (e) at (0, -3.5) {ICFA};

\draw (d) -- (a);
\draw (a) -- (b);
\draw (b) -- (c);
\draw (c) -- (e);

\end{tikzpicture}
\vspace{0.5ex}
\caption{ICFA constr.}
\label{fig:icfa_construction}
\end{minipage}%
\begin{minipage}{0.4\textwidth}
\scriptsize
\centering
\vspace{2ex}
\begin{tikzpicture}
\tikzset{
nodestyle/.style={draw, rounded corners, minimum width=1cm, text width=2.05cm, minimum height=0.75cm, align=center, inner sep=1mm},
edgestyle/.style={->, font=\small}
}

\node[nodestyle] (i) at (0, 0) {dependency analysis};
\node[nodestyle] (a) at (0, -1.875) {pointer\\analysis\makebox[0pt][l]{~**}};
\node[nodestyle] (b) at (2.75, -1.25) {may lockset analysis\makebox[0pt][l]{~*}};
\node[nodestyle] (c) at (2.75, -2.5) {must lockset analysis\makebox[0pt][l]{~*}};
\node[nodestyle] (d) at (5.5, 0) {lock graph construction\makebox[0pt][l]{~*}};
\node[nodestyle] (e) at (5.5, -1.25) {cycle detection};
\node[nodestyle] (f) at (5.5, -2.5) {non-concurrency check};

\node (g) at (7.25, -0.75) {yes};
\node (h) at (7.25, -1.5) {no};

\draw[->] (a) -- (b);
\draw[->] (a) -- (c);
\draw[->] (a) to[out=45, in=180] (d);
\draw[->] (c) -- (f);
\draw[->] (d) -- (e);
\draw[->] (b) -- (d);
\draw[->, dashed] (f) -- (e);
\draw[->] (i) -- (a);

\draw[->] (a) to[out=-90, in=-150] (f);

\draw[->] (e) -- (g);
\draw[->] (e) -- (h);
\end{tikzpicture}
\vspace{2ex}
\caption{Analysis pipeline}
\label{fig:pipeline}
\end{minipage}%
\end{figure*}

The design of our analyses has been guided by the goal to analyse
real-world concurrent C/pthreads code in a \emph{sound} way.
Fig.~\ref{fig:pipeline} gives an overview of our analysis pipeline.
An~arrow between two analyses indicates that the target uses information
computed by the source.  We use a dashed arrow from the non-concurrency
analysis to the cycle detection to indicate that the required information is
computed on-demand (i.e., the cycle detection may repeatedly query the
non-concurrency analysis, which computes the result in a lazy fashion).
All of the analyses operate on a graph representation of the program (introduced
in Sec.~\ref{sec:program_representation}). The exception is the cycle
detection phase, which only uses the lock graph computed in the lock
graph construction phase.

The pointer analysis, may- and must-lockset analysis, and the lock graph
construction are implemented on top of our new generic context- and
thread-sensitive analysis framework (described in detail in
Sec.~\ref{sec:threadsensitive}).  To enable trade-off between precision
and cost, the framework comes in a flow-insensitive and a flow-sensitive
version.  The pointer analysis was implemented on top of the former (thus
marked with ** in Fig.~\ref{fig:pipeline}), and the may- and must-lockset
analysis and the lock graph construction on top of the latter (marked with
*).  The dependency analysis and the non-concurrency analysis are
separate standalone analyses.

\begin{figure}[t]
\scriptsize
\centering
\begin{subfigure}{0.55\linewidth}
\begin{lstlisting}[mathescape,firstnumber=1,xleftmargin=1ex]
int main()
{
  pthread_t tid;
  $$
  pthread_create(&tid, 0,|\label{df:main:create}|
    thread, 0);
  $$
  pthread_mutex_lock(&m1);|\label{df:main:lock1}|
  pthread_mutex_lock(&m3);|\label{df:main:lock2}|
  pthread_mutex_lock(&m2);|\label{df:main:lock3}|
  func1();|\label{df:main:func1}|
  pthread_mutex_unlock(&m2);|\label{df:main:unlock1}|
  pthread_mutex_unlock(&m3);|\label{df:main:unlock2}|
  pthread_mutex_unlock(&m1);|\label{df:main:unlock3}|
  $$
  pthread_join(tid, 0);|\label{df:main:join}|
  $$
  int r;
  r = func2(5)|\label{df:main:func2}|
  $$
  return 0;|\label{df:main:return}|
}|\label{df:main:exit}|
\end{lstlisting}
\begin{lstlisting}[mathescape,firstnumber=22,xleftmargin=1ex]
void func1()
{
  x = 0;|\label{df:func1:x0}|
}
\end{lstlisting}
\end{subfigure}%
\begin{subfigure}{0.55\linewidth}
\begin{lstlisting}[mathescape,firstnumber=26,xleftmargin=1ex]
void *thread()
{
  pthread_mutex_lock(&m1);|\label{df:thread:lock1}|
  pthread_mutex_lock(&m2);|\label{df:thread:lock2}|
  pthread_mutex_lock(&m3);|\label{df:thread:lock3}|
  x = 1;|\label{df:thread:x1}|
  pthread_mutex_unlock(&m3);|\label{df:thread:unlock1}|
  pthread_mutex_unlock(&m2);|\label{df:thread:unlock2}|
  pthread_mutex_unlock(&m1);|\label{df:thread:unlock3}|
  $$
  pthread_mutex_lock(&m4);|\label{df:thread:lock4}|
  pthread_mutex_lock(&m5);|\label{df:thread:lock5}|
  x = 2;|\label{df:thread:x2}|
  pthread_mutex_unlock(&m5);|\label{df:thread:unlock4}|
  pthread_mutex_unlock(&m4);|\label{df:thread:unlock5}|
  $$
  return 0;|\label{df:thread:return}|
}|\label{df:thread:exit}|
\end{lstlisting}
\begin{lstlisting}[countblanklines=false,mathescape,firstnumber=44,xleftmargin=1ex]
int func2(int a)
{
  pthread_mutex_lock(&m5);|\label{df:func2:lock5}|
  pthread_mutex_lock(&m4);|\label{df:func2:lock4}|
  if(a)|\label{df:func2:if}|
    x = 3;|\label{df:func2:x3}|
  else
    x = 4;|\label{df:func2:x4}|
  pthread_mutex_unlock(&m4);|\label{df:func2:unlock4}|
  pthread_mutex_unlock(&m5);|\label{df:func2:unlock5}|
  return 0;|\label{df:func2:return}|
}|\label{df:func2:exit}|
\end{lstlisting}
\end{subfigure}
\vspace{-2ex}
\caption{\label{fig:example}
Example of a deadlock-free program}
\end{figure}

\vspace{-1ex}
\paragraph{Context and thread sensitivity}
Typical patterns in real-world C code suggest that an approach that provides
a form of context-sensitivity is necessary to obtain satisfactory precision
on real-world code, as otherwise there would be too many false deadlock
reports.  For instance, many projects provide their own wrappers for the
functions of the pthreads API.  Fig.~\ref{fig:wrappers}, for example,
shows a lock wrapper from the VLC project.  An analysis that is not
context-sensitive would merge the points-to information for pointer
\texttt{p} from different call sites invoking \texttt{vlc\_mutex\_lock()},
and thus yield many false alarms.

Thread creation causes a similar problem.  For every call to
\texttt{pthread\_create()}, the analysis needs to be determine which thread
is created (i.e., the function identified by the pointer passed to
\texttt{pthread\_create()}).  This is straightforward if a
function identifier is given to \texttt{pthread\_create()}.  However,
similar to the case of lock wrappers above, projects often provide wrappers
for \texttt{pthread\_create()}.  Fig.~\ref{fig:wrappers} shows a wrapper
for \texttt{pthread\_create()} from the memcached project.  The wrapper then
uses the function pointer that is passed to \texttt{create\_worker()} to
create a thread.  Maintaining precision in such cases requires us to track
the flow of function pointer values from function arguments to function
parameters.  This is implemented directly as part of the analysis framework
(as opposed to in the full points-to analysis).






\paragraph{Dependency analysis}

Deadlock detection requires the information which lock objects an expression
used in a \texttt{pthread\_mutex\_lock()} call may refer to.  We compute
this data using the pointer analysis, which is potentially expensive. 
However, it is easy to see that potentially many assignments and function
calls in a program do not affect the values of lock expressions.  Consider
for example Fig.~\ref{fig:example}.  The accesses to \texttt{x} cannot affect
the value of the lock pointers
\texttt{m1}--\texttt{m5}.  Further, the code in function \texttt{func1}
cannot affect the values of the lock pointers, and thus in turn the call
\texttt{func1()} in line~\nolink{\ref{df:main:func1}} cannot affect the the lock pointers.

We have developed a lightweight context-insensitive, flow-insensitive analysis to
identify statements that may affect a given set of expressions. The result is
used to speed up the pointer analysis. The dependency analysis is based on
marking statements which (transitively) share common variables with the
given set of expressions. In our case, the relevant expressions are those used in lock-,
create-, and join-statements. For the latter two we track the thread
ID variable (first parameter of both) whose value is required
to determine which thread is joined by a join operation.
%
%
We~give the details of the dependency analysis in
Sec.~\ref{sec:dependency_analysis}.

\vspace{-1ex}
\paragraph{Non-concurrency analysis}

A deadlock resulting from a thread~1 first acquiring lock~$m_1$ and then
attempting to acquire~$m_2$ (at program location~$\loc_1$), and thread~2 first
acquiring $m_2$ and then attempting to acquire $m_1$ (at program location~$\loc_2$)
can only occur when in a concrete program execution the program locations $\loc_1$
and $\loc_2$ run concurrently. If we have a way of deciding whether two
locations could potentially run concurrently, we can use this information to
prune spurious deadlock reports. For this purpose we have developed a
non-concurrency analysis that can detect whether two statements cannot run
concurrently based on two criteria.

\emph{Common locks}.
If thread 1 and thread 2 hold a common lock at locations $\loc_1$ and
$\loc_2$, then they cannot both simultaneously reach those locations, and
hence the deadlock cannot happen.  This is illustrated in
Fig.~\ref{fig:example}.  The thread \texttt{main()} attempts to acquire the locks
in the sequence $m_1$, $m_3$, $m_2$, and the thread \texttt{thread()} attempts to
acquire the locks in the sequence $m_1$, $m_2$, $m_3$.  There is an order
inversion between $m_2$ and $m_3$, but there is no deadlock since the two
sections \nolink{\ref{df:main:lock1}}--\nolink{\ref{df:main:unlock3}} and
\nolink{\ref{df:thread:lock1}}--\nolink{\ref{df:thread:unlock3}} (and thus in particular the locations \nolink{\ref{df:main:lock3}} and
\nolink{\ref{df:thread:lock3}}) are protected by the common lock $m_1$.  The common locks criterion
has first been described by Havelund~\cite{Hav00} (common locks are called
\emph{gatelocks} there).

\emph{Create and join}. Statements might also not be able to run
concurrently because of the relationship between threads due to the \texttt{pthread\_create()}
and \texttt{pthread\_join()} operations.  In Fig.~\ref{fig:example}, there is an order
inversion between the locks of $m_5$ and $m_4$ by function \texttt{func2()}, and the
locks of $m_4$, $m_5$ of thread \texttt{thread()}.  Yet there is no deadlock since
the thread \texttt{thread()} is joined before \texttt{func2()} is invoked.

Our non-concurrency analysis makes use of the must lockset analysis (computing
the locks that must be held) to detect common locks. To detect the relationship
between threads due to create and join operations it uses a search on the
program graph for joins matching earlier creates.
We give more details of our non-concurrency analysis in
Sec.~\ref{sec:nonconcurrent}.

\section{Analysis Framework}\label{sec:analysis_framework}

In this section, we first introduce our program representation, then describe
our context- and thread-sensitive framework, and then describe the
pointer analysis and lockset analyses that are implemented on top of the
framework.

\subsection{Program Representation}\label{sec:program_representation}

\emph{Preprocessing}.
Our tool takes as input a concurrent C program using the pthreads threading
library.
In a first step the calls to functions through function pointers are removed. A
call is replaced by a case distinction over the functions the function pointer
could refer to. Specifically, a function pointer can only refer to functions
that are type-compatible and of which the address has been taken at some point
in the code. This is illustrated in Fig.~\ref{fig:function_pointers} (top).
Functions \texttt{f2()} (address not taken) and \texttt{f4()} (not
type-compatible) do not have to be part of the case distinction.
In a second step, functions with multiple exit points (i.e., multiple return
statements) are transformed such as to have only one exit points (illustrated in
Fig.~\ref{fig:function_pointers} (bottom)).

\emph{Interprocedural CFAs}. We transform the program into a graph
representation which we term \emph{interprocedural control flow automaton}
(ICFA). The functions of the program are represented as CFAs~\cite{HJMS02}. CFAs
are similar to control flow graphs, but with the nodes representing program
locations and the edges being labeled with operations. ICFAs have additional
inter-function edges modeling function entry, function exit, thread entry,
thread exit, and thread join. Fig.~\ref{fig:example} shows a concurrent C
program and Fig.~\ref{fig:cfa} shows its corresponding ICFA (leaving off
thread exit and thread join edges and the function \texttt{func1()}).



\begin{figure*}[t]
\scriptsize
\begin{minipage}{0.37\textwidth}
\begin{minipage}{\textwidth}
\begin{minipage}{0.37\textwidth}
\begin{lstlisting}
void f1() {}
void f2() {}
void f3() {}
int f4() {}
...
... = &f1;
... = &f3;
... = &f4;
fp();
\end{lstlisting}
\end{minipage}%
\begin{minipage}{0.12\textwidth}
\large
$\Rightarrow$
\end{minipage}%
\begin{minipage}{0.45\textwidth}
\begin{lstlisting}
...
if(fp==f1)
 f1();
else
if(fp==f3)
 f3();
\end{lstlisting}
\end{minipage}
\end{minipage}
\begin{minipage}{\textwidth}
\begin{minipage}{0.37\textwidth}
\begin{lstlisting}
int f() {
 if(...)
   return 0;
 else
   return 1;
}
...
a = f();
\end{lstlisting}
\end{minipage}%
\begin{minipage}{0.12\textwidth}
\large
$\Rightarrow$
\end{minipage}%
\begin{minipage}{0.5\textwidth}
\begin{lstlisting}
int f() {
  int ret;
  if(...)
    ret = 0;
    goto END;
  else
    ret = 1;
    goto END;
END:
  return ret;
}
...
a = f();
\end{lstlisting}
\end{minipage}
\end{minipage}
\caption{Function pointers and returns}
\label{fig:function_pointers}
\end{minipage}%
\begin{minipage}{0.6\textwidth}
\begin{tikzpicture}

\tiny

\tikzset{
nodestyle/.style={draw, circle, minimum size=4.5mm, inner sep=0.5mm},
edgestyle/.style={->, font=\small},
cl/.style={midway, anchor=right, left},
cr/.style={midway, anchor=left, right}
}

\begin{scope}
\node[yshift=0.7cm] (fs) {\small func2(int a)};
\node[nodestyle] (f0) {$\nolink{\ref{df:func2:lock5}}$};
\node[nodestyle, below of=f0] (f1) {$\nolink{\ref{df:func2:lock4}}$};
\node[nodestyle, below of=f1] (f2) {$\nolink{\ref{df:func2:if}}$};

\node[nodestyle, below of=f2, xshift=-1cm] (f3) {$\nolink{\ref{df:func2:x3}}$};
\node[nodestyle, below of=f2, xshift=1cm] (f4) {$\nolink{\ref{df:func2:x4}}$};

\node[nodestyle, below of=f4, xshift=-1cm] (f5) {$\nolink{\ref{df:func2:unlock4}}$};
\node[nodestyle, below of=f5] (f6) {$\nolink{\ref{df:func2:return}}$};
\node[nodestyle, below of=f6] (f7) {$\nolink{\ref{df:func2:exit}}$};

\draw[edgestyle] (f0) -- (f1) node[cl] {\lockop($m_5$)};
\draw[edgestyle] (f1) -- (f2) node[cl] {\lockop($m_4$)};

\draw[edgestyle] (f2) -- (f3) node[cl] {$[a\neq 0]$};
\draw[edgestyle] (f2) -- (f4) node[cr] {$[a = 0]$};

\draw[edgestyle] (f3) -- (f5) node[cl] {\texttt{x = 3}};
\draw[edgestyle] (f4) -- (f5) node[cr] {\texttt{x = 4}};

\node (e1) at (0, -4.5) {\small...};
\draw[edgestyle] (f6) -- (f7) node[cl] {return 0};

\end{scope}

\begin{scope}[xshift=4cm]
\node[yshift=0.7cm] (fs) {\small main()};
\node[nodestyle] (m0) {$\nolink{\ref{df:main:create}}$};
\node[nodestyle, below of=m0] (m1) {$\nolink{\ref{df:main:lock1}}$};
\node[nodestyle, below of=m1] (m2) {$\nolink{\ref{df:main:lock2}}$};

\node[nodestyle, below of=m2] (m3) {$\nolink{\ref{df:main:join}}$};
\node[nodestyle, below of=m3] (m4) {$\nolink{\ref{df:main:func2}}$};
\node[nodestyle, below of=m4] (m5) {$\nolink{\ref{df:main:return}}$};
\node[nodestyle, below of=m5] (m6) {$\nolink{\ref{df:main:exit}}$};

\draw[edgestyle] (m0) -- (m1) node[cr, align=left, yshift=-2ex] {\createop(\&tid, \texttt{0},\\\hspace{6.5ex} thread, \texttt{0})};
\draw[edgestyle] (m1) -- (m2) node[cr] {\lockop($m_3$)};

\draw[edgestyle] (m3) -- (m4) node[cr] {\joinop(tid, \texttt{0})};
\draw[edgestyle] (m4) -- (m5) node[cr] {r = func2(3)};
\draw[edgestyle] (m5) -- (m6) node[cr] {return 0};
\end{scope}

\begin{scope}[xshift=9cm]
\node[yshift=0.7cm] (fs) {\small thread()};
\node[nodestyle] (t0) {$\nolink{\ref{df:thread:lock1}}$};
\node[nodestyle, below of=t0] (t1) {$\nolink{\ref{df:thread:lock2}}$};
\node[nodestyle, below of=t1] (t2) {$\nolink{\ref{df:thread:lock3}}$};

\node[nodestyle, below of=t2] (t3) {$\nolink{\ref{df:thread:unlock4}}$};
\node[nodestyle, below of=t3] (t4) {$\nolink{\ref{df:thread:unlock5}}$};
\node[nodestyle, below of=t4] (t5) {$\nolink{\ref{df:thread:return}}$};
\node[nodestyle, below of=t5] (t6) {$\nolink{\ref{df:thread:exit}}$};

\draw[edgestyle] (t0) -- (t1) node[cl] {\lockop($m_1$)};
\draw[edgestyle] (t1) -- (t2) node[cl] {\lockop($m_2$)};

\draw[edgestyle] (t3) -- (t4) node[cl] {\unlockop($m_5$)};
\draw[edgestyle] (t4) -- (t5) node[cl] {\unlockop($m_4$)};
\draw[edgestyle] (t5) -- (t6) node[cl] {return \texttt{0}};
\end{scope}

\draw[edgestyle, dashed] (m0) -- (t0) node[above, midway]
  {\threadentry($\text{thread}$, \texttt{0}, par)};
\draw[edgestyle, dashed] (m4) -- (f0) node[above=5mm, right=-4mm, midway] {\funcentry(5, a)};
\draw[edgestyle, dashed] (f7) -- (m5) node[above=2mm, midway] {\funcexit(0, r)};

\node (e1) at (3.98, -2.5) {\small...};
\node (e2) at (8.98, -2.5) {\small...};

\end{tikzpicture}
\caption{ICFA associated with the program in Fig.~\ref{fig:example}}
\label{fig:cfa}
\end{minipage}
\end{figure*}

We denote by $\prog$ a program (represented as an ICFA), by $\funcs$ the set of
identifiers of the functions,
by $\locs = \{\loc_0,
\ldots, \loc_{n-1}\}$ the set of program locations, by $\edges$ the
set of edges connecting the locations, and by $\op(e)$ a function that
labels each edge with an
operation. For example, in Fig.~\ref{fig:cfa}, the edge between locations
$\nolink{\ref{df:func2:x3}}$ and $\nolink{\ref{df:func2:unlock4}}$ is labeled with the operation \texttt{x=3}, and the edge between
locations $\nolink{\ref{df:main:func2}}$ and $\nolink{\ref{df:func2:lock5}}$ is labeled with the operation \funcentry(5, a).
%
%

We further
write $\locs(f)$ for the locations in function $f$.
Each program location
is contained in exactly one function. The function $\getfunc(\loc)$ yields
the function that contains~$\loc$.  The set of variable identifiers in the
program is denoted by~$\vars$. We assume that all identifiers in $\prog$ are
unique, which can always be achieved by a suitable renaming of identifiers.
%
%

We treat lock, unlock, thread create, and thread join as primitive
operations.  That is, we do not analyse the body of
e.g.~\texttt{pthread\_create()} (as implemented in e.g.~glibc on GNU/Linux
systems).  Instead, our analysis only tracks the semantic effect of the operation,
i.e., creating a new thread.

Apart from intra-function edges we also have inter-function edges that can be
labeled with the five operations \funcentry, \funcexit, \threadentry,
\threadexit, and \threadjoin.

A function entry edge (\funcentry) connects a call site to the function entry
point. The edge label also includes the
function call arguments and the function parameters. For example, \funcentry(5, a)
indicates that the integer literal~$5$ is passed to the call as an argument,
which is assigned to function parameter~$a$. A~function exit edge
(\funcexit) connects the exit point of a function to
\emph{every} call site calling the function.  Our analysis algorithm filters
out infeasible edges during the exploration of the ICFA.  That is, if a
function entry edge is followed from a function $f_1$ to function $f_2$,
then the analysis algorithm later follows the exit edge from $f_2$ to $f_1$,
disregarding exit edges to other functions.

A thread entry edge (\threadentry) connects a thread creation site to all
potential thread entry points.
It is necessary to connect to all potential thread entry points since often a
thread creation site can create threads of different types (i.e., corresponding
to different functions), depending on the value of the function pointer passed
to \texttt{pthread\_create()}. Analogous to the case of function exit edges, our analysis
algorithm tracks the values of function pointers during the ICFA exploration. At
a thread creation site it thus can resolve the function pointer, and only
follows the edge to the thread entry point corresponding to the value of the
function pointer.

A $\threadexit$ edge connects the exit point of a thread to the location following
all thread creation sites, and a $\threadjoin$ edge connects a thread exit point
to all join operations in the program.









\subsection{Analysis Framework -- Overview}\label{sec:threadsensitive}


\begin{figure}[h!]
\normalsize

\hrulefill
\vspace{1ex}

Domain: $\doms = \fpms \times \doma$

\vspace{-0.5ex}
\hrulefill
\vspace{-4ex}


\begin{align*}
\states^{1} \joins \states^{2} = & (\fpm, \statea)\\
& \text{with}~\states^{1} = (\fpm^1, \statea^{1})\\
& \text{with}~\states^{2} = (\fpm^2, \statea^{2})\\
& \text{with}~\fpm = \fpm^1 \joinfp \fpm^2\\
& \text{with}~\statea = \statea^{1} \joina \statea^{2}
\end{align*}





\hspace{12.5ex}$\fpm^1 \joinfp \fpm^2 = \fpm$\\

\vspace{-1.5ex}

$\fpm(\fp) =$
\[
\begin{cases}
d & \fpm^1(\fp) = d \vee \fpm^2(\fp) = d\\
v & (\fpm^1(\fp) = v \wedge (\fp \notin \dom(\fpm^2) \vee \fpm^2(\fp) = v)) \vee\\
  & (\fpm^2(\fp) = v \wedge (\fp \notin \dom(\fpm^1) \vee \fpm^1(\fp) = v))\\
\bot & \text{otherwise}
\end{cases}
\]



\vspace{-1.5ex}
\hrulefill
\vspace{1ex}

With $\cfaedge = (\loc_1, \loc_2)$, $\placetop(\place) = \loc_1$, $f = \getfunc(\loc_2)$, and $n = |p|$:\\

\vspace{-1ex}

$\placenexts(\cfaedge, \place) =$
\[
\begin{cases}
  \entry_s(\place, \loc_2) &\op(\cfaedge) \in \{\threadentry, \funcentry\}\\
  \place[:n-2] + \loc_2    &\op(\cfaedge) \in \{\funcexit, \threadexit, \threadjoin\}\\
  \place[:n-1] + \loc_2    &\text{otherwise}
\end{cases}
\]

\vspace{1ex}

$\entry_s(\place, \loc) =$
\[
\begin{cases}
\place' + \loc' + \loc & \place = \place' + \loc' + \loc'' + \place'' \wedge \getfunc(\loc'') = f\\
\place  + \loc  & \text{otherwise}
\end{cases}
\]



\vspace{-0.5ex}
\hrulefill
\vspace{1ex}

With $\cfaedge = (\loc_{src}, \loc_{tgt})$, $\op(e) = \funcentry(\arg_1,\ldots,\arg_k, \param_1, \ldots, \param_k)$, and $\states = (\fpm, \statea)$:

\vspace{-2.5ex}

\begin{align*}
\transfs{e, \place}(\states) = & (\fpm', \transfa{e, \place\,}(\statea))
\end{align*}

\vspace{-2.5ex}

\[
\fpm'(\param_i) =
\begin{cases}
\arg_i & \isfunc(\arg_i)\\
\fpm(\arg_i) & \isfp(\arg_i)\\ 
\end{cases}
\]

\vspace{-1ex}
\hrulefill
\vspace{1ex}

With $\cfaedge = (\loc_{src}, \loc_{tgt})$, $f = \getfunc(\loc_{tgt})$,
$\op(\cfaedge) = \threadentry(\mathit{thr}, \arg, \param)$, and $\states = (\fpm, \statea)$:\\


$
\transfs{e, p}(\states) =
\begin{cases}
(\fpm', \transfa{e, \place}(\statea)) & \matchfp(\fpm, \mathit{thr}, f)\\
(\emptyset, \bota) & \text{otherwise}\\
\end{cases}
$

\vspace{-1ex}\begin{align*}
&\matchfp(\fpm, \mathit{thr}, f) =\\
&~~(\isfp(\mathit{thr}) \wedge \fpm(\mathit{thr}) \in \{\bot, d, \func\})\,\vee\\
&~~(\isfunc(\mathit{thr}) \wedge \mathit{thr} = \func)
\end{align*}

%



\vspace{-2ex}
\hrulefill
\vspace{1ex}

With $\cfaedge = (\loc_{src}, \loc_{tgt})$, $\op(e) \in \{\funcexit, \threadexit, \threadjoin\}$, and $\states = (\fpm, \statea)$:\vspace{2ex}\\

\vspace{-3ex}

$
\transfs{\cfaedge, \place}(\states) =
(\emptyset, \transfa{\cfaedge, \place}(\statea))
$

\vspace{0ex}
\hrulefill
\vspace{1ex}

With $\cfaedge = (\loc_{src}, \loc_{tgt})$, $\op(e) = \op$, and $\states = (\fpm, \statea)$:\\

\vspace{-1ex}

$
\transfs{\cfaedge, \place}(\states) = (\fpm, \transfa{\cfaedge, \place}(\statea))
$

\vspace{-0.5ex}
\hrulefill

\vspace{-1ex}
\caption{\label{fig:framework}
Context-, thread-, and flow-sensitive framework}
\end{figure}

Our framework to perform context- and thread-sensitive
analyses on ICFAs is based on abstract
interpretation~\cite{CC77}. It implements a flow-sensitive and
flow-insensitive fixpoint computation over the ICFA, and needs to be
parametrised with a custom analysis to which it delegates the handling of
individual edges of the ICFA. We provide more details and a formalization of the
framework in the next section.



Our analysis framework unifies contexts, threads,
and program locations via the concept of a \emph{place}.
%
%
A place is a tuple $(\loc_0, \loc_1, \ldots, \loc_n)$ of program
locations. The program locations $\loc_0, \ldots, \loc_{n-1}$ are
either function call sites or thread creation sites in the
program (such as, e.g., location~\nolink{\ref{df:main:func2}} in Fig.~\ref{fig:cfa}). The final location $\loc_n$ can be a program location of any
type. The locations $\loc_0, \ldots, \loc_{n-1}$ model a possible
function call and thread creation history that leads up to program
location $\loc_n$.
We denote the set of all places for a given program by $\places$. We use the $+$
operator to extend tuples, i.e., $(\loc_0, \ldots, \loc_{n-1}) + \loc_n =
(\loc_0, \ldots, \loc_{n-1}, \loc_n)$.
We further write $|\place|$ for the
length of the place.
We write $\place[i]$ for element $i$ (indices are 0-based).
We use slice notation to refer to contiguous parts of
places; $\place[i\colon\!j]$ denotes the part from index $i$ (inclusive) to
index $j$ (exclusive), and $\place[\colon\!i]$ denotes the prefix until index $i$ (exclusive).
We write $\placetop(\place)$ for the last location in the place.

As an example, in Fig.~\ref{fig:cfa}, place $(\nolink{\ref{df:main:func2}}, \nolink{\ref{df:func2:x3}})$
denotes the program location $\nolink{\ref{df:func2:x3}}$ in function \texttt{func2()} when it has
been invoked at call site $\nolink{\ref{df:main:func2}}$ in the main function. If function
\texttt{func2()} were called at multiple program locations $\loc_1,\allowbreak \ldots,\allowbreak
\loc_m$ in the main function, we would have different places $(\loc_1, \nolink{\ref{df:func2:x3}}),
\ldots, (\loc_m, \nolink{\ref{df:func2:x3}})$ for location $\nolink{\ref{df:func2:x3}}$ in function
\texttt{func2()}. Similarly, for the thread function \texttt{thread()} and, e.g., location $\nolink{\ref{df:thread:lock2}}$, we have a place $(\nolink{\ref{df:main:create}}, \nolink{\ref{df:thread:lock2}})$ with $\nolink{\ref{df:main:create}}$ identifying the
creation site of the thread.

Each place has an associated abstract thread identifier, which we refer
to as \emph{thread I\hspace{-0.5pt}D} for short. Given a place $\place =
(\loc_0, \ldots, \allowbreak \loc_{n})$, the associated thread ID is either $t =
()$ (the empty tuple) if no location in $\place$ corresponds to a thread
creation site, or $t = \loc_0, \ldots, \loc_i$, such that $\loc_i$ is a thread
creation site and all $\loc_j$ with $j > i$ are not thread creation sites.
%
%
It is in this
sense that our analysis is thread-sensitive, as the information
computed for each place can be associated with an abstract thread that
way. We write $\getthread(\place)$ for the thread ID associated with place $\place$.


The analysis framework must be parametrised with a custom analysis. The framework
handles the tracking of places, the tracking of the flow of function pointer values from function
arguments to function parameters, and it
invokes the custom analysis to compute dataflow facts for each
place.

The domain, transfer function, and join function of the
framework are denoted by $\doms$, $\transfsf$, and $\joins$, respectively, and
the domain, transfer function, and join function of the parametrising
analysis are denoted by $\doma$, $\transfaf$, and $\joina$.
%
%
The custom analysis has a transfer function $\transfaf: \cfaedges \times \places \rightarrow (\doma
\rightarrow \doma)$ and a join function $\joina: \doma \times \doma
\rightarrow \doma$. The domain of the framework (parametrised by the
custom analysis) is then $\doms = \fpms \times \doma$, the transfer
function is $\transfsf: \cfaedges \times \places \rightarrow (\doms \rightarrow \doms)$, and the join
function is $\joins: \doms \times \doms \rightarrow \doms$.

The set $\fpms$ is a
set of mappings from identifiers to functions which map function pointers to the
functions they point to. We denote the empty mapping by $\emptyset$. We further
write $\fpm(\fp) = \bot$ to indicate that $\fp$ is not in $\dom(\fpm)$ (the
domain of $\fpm$). A function pointer $\fp$ might be mapped by $\fpm$ either to
a function $f$ or to the special value $d$ (for ``dirty'') which indicates that
the analysed function assigned to $\fp$ or took the address of $\fp$. In this
case we conservatively assume that the function pointer could point to any
thread function.

\subsection{Analysis Framework -- Details}\label{sec:framework_details}

We now describe the formalization of the analysis framework which is shown in
Fig.~\ref{fig:framework}. The figure gives the flow-sensitive variant of our
framework. We refer to
\rronly{Appendix~\ref{sec:framework_fi}}
\pponly{the extended version of the paper~\cite{kroening:2016}}
for the flow-insensitive version.
The figure gives the domain, join
function $\joins$ and transfer function $\transfsf$ which are defined in terms
of the join function $\joina$ and transfer function $\transfaf$ of the
parametrising analysis (such as the lockset analyses defined in the next
section).

The function
$\placenexts(e, \place)$ defines how the place $\place$ is updated when the
analysis follows the ICFA edge~$e$.
For example, on a $\funcexit$ edge, the last two locations are removed from the
place (which are the exit point of the function, and the location of the call to
the function), and the location to which the function returns to is added to the
place (which is the location following the call to the function). The
$\threadentry$ and $\funcentry$ cases are delegated to $\textsf{entry}_s(\place,
\loc)$. The first case of the function handles recursion. If a location
$\loc''$ of the called function is already part of the place, then the prefix of
the place that corresponds to the original call to the function is reused (first
case). If no recursion is detected, the entry location of the function is simply
added to the current place (second case). For intra-function edges (last case of
$\placenext_s$), the last location is removed from the place and the target
location of the edge is added.


The overall result of the analysis is a mapping $s \in P
\rightarrow (\fpms \times \doma)$. The result is defined via a fixpoint
equation~\cite{CC77}. We obtain the result by computing the least
fixpoint (via a worklist algorithm) of the equation below (with $s_0$
denoting the initial state of the places):



\hspace{-4ex}
\begin{tikzpicture}
\node[right] (a) at (0, 0) {
\begin{minipage}{\linewidth}
\begin{flalign*}
s = s_0~\join~&\lambda\,p. \bigsqcup_{p', e~\text{s.t.} \textsf{np}(p, p', e)} \transfs{e, p'}(s(p'))
\end{flalign*}
\end{minipage}
};
\node (c) at (4.54, -0.18) {\scriptsize$c$};
\end{tikzpicture}

\vspace{-5ex}

\begin{flalign*}
\text{with}~~\textsf{np}(p, p', (\loc_1, \loc_2))~=~&\loc_1 = \textsf{top}(p') \wedge\\
                                         ~&\loc_2 = \textsf{top}(p) \wedge\\
                                         ~&\placenexts((\loc_1, \loc_2), p') = p\\[1ex]
\text{with}~~s \join s'~=~&\lambda\, p.\, s(p) \join_c s'(p)
\end{flalign*}


\noindent
The equation involves computing the join over all places~$\place'$ and edges
$e$ in the ICFA such that $\textsf{np}(\place, \place', e)$.

We next describe the definition of the transfer function of the framework in
more detail. The definition consists of four cases: (1) function entry, (2)
thread entry, (3) function exit, thread exit, thread join, and (4)
intra-function edges.

(1) When applying a function entry edge, a new function pointer map $\fpm'$ is
created by assigning arguments to parameters and looking up the values of the
arguments in the current function pointer map $\fpm$. As in the following cases,
the transfer function $\transfaf$ of the custom analysis is applied to the state
$\statea$.


(2) Applying a thread entry edge to a state $s_c$ yields one of two outcomes.
When the value of the function pointer argument $\mathit{thr}$ matches the
target of the edge (i.e., the edge enters the same function as the function
pointer points to), then the function pointer map is updated with $\arg$ and
$\param$ (as in the previous case), and the transfer function of the custom
analysis is applied. Otherwise, the result is the bottom element $\bot_c =
(\emptyset, \bot_a)$.


(3) The function pointer map is cleared (as its domain contains only parameter
identifiers which are not accessible outside of the function), and the custom
transfer function is applied.


(4) The custom transfer function is applied.

As is not shown for lack of space in Fig.~\ref{fig:framework}, if a function
pointer $\fp$ is assigned to or its address is taken, its value is set to $d$ in
$\fpm$, thus indicating that it could point to any thread function.

\paragraph{Implementation}
During the analysis we need to keep a mapping from places to abstract states (which
we call the \emph{state map}).
However, directly using the places as keys for the state maps in all analyses can lead
to high memory consumption. Our implementation therefore keeps a global two-way
mapping (shared by all analyses in Fig.~\ref{fig:pipeline}) between places and
unique IDs for the places (we call this the \emph{place map}). The state maps of
the analyses are then mappings from unique IDs to abstract states, and the
analyses consult the global place map to
translate between places and IDs when needed.

In the two-way place map, the mapping from places to IDs is implemented via a
trie, and the mapping from IDs to places via an array that stores pointers back
into the trie. The places in a program can be efficiently stored in a trie as
many of them share common prefixes. We give further details in
\rronly{Figure~\ref{fig:trie_numbering} in Appendix~\ref{sec:implementation}.}
\pponly{the extended version~\cite{kroening:2016}.}
%

\begin{figure}[t]
  \footnotesize

\hrulefill

Domain: $2^{\objs} \cup \{\{\unknownlock\}\}$

\hrulefill

$s_1 \join s_2 =
\begin{cases}
s_1 \cup s_2 & \text{if}~s_1, s_2 \neq \{\unknownlock\}\\
\{\unknownlock\}     & \text{otherwise}\\
\end{cases}$

\hrulefill



With $\op(e) = \lockop(a)$:

$
\transf{e, p}(s) =
\begin{cases}
s \cup vs(p, a) & \text{if}~s, vs(p, a) \neq \{\unknownlock\}\\
\{\unknownlock\}        & \text{otherwise}\\
\end{cases}
$

\vspace{2ex}

With $\op(e) = \unlockop(a)$:

$
\transf{e, p}(s) =
\begin{cases}
\emptyset    & \text{if}~|s|=1 \wedge s \neq \{\unknownlock\}\\
s - vs(p, a) & \text{if}~|s \cap vs(p, a)|=1 \wedge\\
             & \hspace{2.5ex}s \neq \{\unknownlock\} \wedge vs(p, a) \neq \{\unknownlock\}\\
s            & \text{otherwise}
\end{cases}
$

\vspace{2ex}

With $\op(e) \in \{\threadentry, \threadexit, \threadjoin\}$:

\vspace{1.5ex}

$
\transf{e, p}(s) = \emptyset
$

\vspace{0ex}
\hrulefill

\vspace{-1ex}

\caption{\label{fig:may_lock_set_analysis}May lockset analysis}
\end{figure}
\subsection{Pointer Analysis}\label{sec:alias}

We use a standard points-to analysis that is an instantiation of the flow-insensitive
version of the above framework 
\rronly{(see Appendix~\ref{sec:framework_fi}).}
\pponly{(see the extended version of the paper~\cite{kroening:2016}).}
It computes for each place an element of $\vars \rightarrow (2^{\objs} \cup \{\{\unknownlock\}\})$.
That is, the set of possible values of a pointer variable is either a finite set
of objects it may point to, or $\{\unknownlock\}$ to indicate that it could point to any
object. We~use $vs(\place, a)$ to denote the value set at place~$p$ of
pointer~$a$.
The pointer analysis is sound for concurrent programs due to its
flow-insensitivity~\cite{rinard:2001}.

%



\subsection{Lockset Analysis}\label{sec:lockset}

Our analysis pipeline includes a may lockset analysis (computing for each
place the locks that may be held) and a must lockset analysis (computing for
each place the locks that must be held). The former is used by the lock graph
analysis, and the latter by the non-concurrency analysis.

The may lockset analysis is formalised in Fig.~\ref{fig:may_lock_set_analysis}
as a custom analysis to parametrise the flow-sensitive framework with. The
must lockset analysis is given in 
\rronly{Appendix~\ref{sec:must_lockset_analysis_app}.}
\pponly{the extended version~\cite{kroening:2016}.}
%
%
%
Both the may and must lockset analyses makes use of the
previously computed points-to information by means of the function $vs()$. In
both cases, care
must be taken to compute sound information from the
may-point-to information provided by $vs()$. For example, for the may lockset
analysis on an $\unlockop(a)$ operation, we cannot just remove all elements in
$vs(p, a)$ from the lockset, as an unlock can only unlock one lock.
We use $\maylockset(\place), \mustlockset(\place)$ to denote the may and
must locksets at place~$p$.

\section{Dependency Analysis}\label{sec:dependency_analysis}

We have developed a context-insensitive, flow-insensitive \emph{dependency
analysis} to compute the set of assignments and function calls that might affect
the value of a given set of expressions (in our case the expressions used in lock-,
\mbox{create-,} and join-statements). The purpose of the analysis is to speed up the
following pointer analysis phase (cf. Fig.~\ref{fig:pipeline}).

Below we first describe a semantic characterisation
of dependencies between expressions and assignments, and then devise an
algorithm to compute dependencies based on syntax only (specifically, the
variable identifiers occuring in the expressions/assignments).

\paragraph{Semantic characterisation of dependencies}

Let $\mathit{AS} = \{ e \in \cfaedges(\prog)~|~\textsf{is\_assign}(\op(e))\}$ be the set of assignment edges. Let $\mathit{exprs}$
be a set of starting expressions. Let further $R(a), W(a)$ denote the set of memory
locations that an expression or assignment $a$ may read (resp.~write) over
\emph{all} possible executions of the program. Let further $M(a) = R(a) \cup
W(a)$. Then we define the immediate dependence relation $\mathit{dep}$ as follows
(with $*$ denoting transitive closure and $;$ denoting composition):

\vspace{-2.5ex}

\begin{align*}
&\mathit{dep}_1 \subseteq \mathit{exprs} \times \mathit{AS}, (a, b) \in \mathit{dep}_1 \Leftrightarrow R(a) \cap W(b) \neq \emptyset\\
&\mathit{dep}_2 \subseteq \mathit{AS} \times \mathit{AS}, (a, b) \in \mathit{dep}_2 \Leftrightarrow R(a) \cap W(b) \neq \emptyset\\
&\mathit{dep} = \mathit{dep}_1; \mathit{dep}_2^*
\end{align*}

If $(a, b) \in \mathit{dep}_1$, then the evaluation of expression~$a$
may read a memory location that is written to by assignment~$b$.  If~$(a, b)
\in \mathit{dep}_2$, then the evaluation of the assignment $a$ may read a
memory location that is written to by the assignment~$b$.  If $(a, b) \in
\mathit{dep}$, this indicates that the expression~$a$ can (transitively) be
influenced by the assignment~$b$.  We~say $a$ \emph{depends on} $b$ in this
case.

The goal of our dependency analysis is to compute the set of assignments $A =
\mathit{dep}|_{(\_, a) \mapsto a}$ (the binary relation $A$ projected to the
second component). 
However, we cannot directly implement a procedure based on the
definitions above as this would require the functions $R()$, $W()$ to return
the memory locations accessed by the expressions/assignments. This in turn would
require a pointer analysis--the very thing we are trying to optimise.

Thus, in the next section, we outline a procedure for computing the relation
$\mathit{dep}$ which relies on the symbols (i.e., variable identifiers) occuring
in the expressions/assignments rather then the memory locations accessed by
them.

\paragraph{Computing dependencies}

In this section we outline how we can compute an overapproximation of the set
of assignments $A$ as defined above. Let $\symbols(a)$ be a function that returns
the set of variable identifiers occuring in an expression/assignment.
For example,
$\symbols(\texttt{a[i]->lock}) = \{\texttt{a}, \texttt{i}\}$ and
$\symbols(\texttt{*p=q+1}) = \{\texttt{p}, \texttt{q}\}$.
As stated
in Sec.~\ref{sec:program_representation}, in our program representation all
variable identifiers in a program are unique.
We~first define the relation
$\mathit{sym}_2$ which indicates whether two assignments have common symbols:

\vspace{-2.5ex}

\begin{align*}
&\mathit{sym}_2 \subseteq \mathit{AS} \times \mathit{AS}\\
&(a, b) \in \mathit{sym}_2 \Leftrightarrow \symbols(a) \cap \symbols(b) \neq \emptyset
\end{align*}

Our analysis relies on the following property: If two assignments $a, b$ can
access a common memory location (i.e., $M(a) \cap M(b) \neq \emptyset$), then
$(a, b) \in \mathit{sym}_2^*$.
This can be seen as follows. Whenever a memory region/location is allocated in C
it initially has at most one associated identifier. For example, the memory
allocated for a global variable \texttt{x} at program startup has initially just
the associated identifier~\texttt{x}. Similarly, memory allocated via, e.g.,
\texttt{a = (int *)malloc(sizeof(int) * NUM)} has initially only the associated
identifier \texttt{a}. If an expression not mentioning \texttt{x}, such as
\texttt{*p}, can access the associated memory location, then the address of
\texttt{x} must have been propagated to \texttt{p} via a sequence of assignments
such as \texttt{q=\&x},  \texttt{s->f=q}, \texttt{p=s->f}, with each of the adjacent
assignments having common variables. Thus, if $a, b$ can access a common memory
location, then both must be ``connected'' to the initial identifier associated
with the location via such a sequence. Thus, in particular, $a, b$ are also
connected. Therefore, $(a, b) \in \mathit{sym}_2^*$.

We next define the $\mathit{sym}$ relation which also incorporates the starting
expressions:

\vspace{-2.5ex}

\begin{align*}
&\mathit{sym}_1 \subseteq \mathit{exprs} \times \mathit{AS}\\
&(a, b) \in \mathit{sym}_1 \Leftrightarrow \symbols(a) \cap \symbols(b) \neq \emptyset\\
&\mathit{sym} = \mathit{sym}_1; \mathit{sym}_2^*
\end{align*}

\noindent
As we will show below we have $\mathit{dep} \subseteq \mathit{sym}$ and thus
also $A = \mathit{dep}|_{(\_, a) \mapsto a} \subseteq \mathit{sym}|_{(\_, a)
\mapsto a}$.  Thus, if we compute $\mathit{sym}$ above we get an
overapproximation of~$A$.

The fact that $\mathit{dep} \subseteq \mathit{sym}$ can be seen as follows.
Let $(a, b) \in \mathit{dep}$.  Then there are $a_1, a_2, \ldots, a_n, b$
such that $(a_1, a_2)\allowbreak \in\allowbreak \mathit{dep}_1 \cup
\mathit{dep}_2, (a_2, a_3) \in \mathit{dep}_2, \ldots, (a_n, b) \in
\mathit{dep}_2$.  Let $(a', a'')$ be an arbitrary one of those pairs.  Then
$R(a) \cap W(b) \neq \emptyset$ by the definition of $\mathit{dep}_1$ and
$\mathit{dep}_2$.  Thus $M(a) \cap M(b) \neq \emptyset$.  As we have already
argued above, if two expressions/assignments can access the same memory
location then they must transitively share symbols.  Thus $(a', a'') \in
\mathit{sym}_1 \cup \mathit{sym}_2^*$ must hold.  Therefore, since we have
chosen $(a', a'')$ arbitrarily, we have that all of the pairs above are
contained in $\mathit{sym}_1 \cup \mathit{sym}_2^*$ and thus by the
definition of $\mathit{sym}$ and in particular the transitivity of
$\mathit{sym}_2^*$ we get $(a, b) \in \mathit{sym}$.

Thus, we can use the definition of $\mathit{sym}$ above to compute an
overapproximation of the set of assignments that can affect the starting
expressions as defined semantically in the previous section.

\paragraph{Algorithm}

Algorithm~\ref{alg:dependency_analysis} gives our dependency analysis. The
first phase (line 1, Algorithm~\ref{alg:affecting_assignments}) is based on
the ideas from the previous section and computes the set of assignments~$A$
that can affect the given set of starting expressions $\mathit{exprs}$.  It
does so by keeping a global set of variable identifiers $I$ which is
initialised to the set of variables occuring in the starting expressions. 
Then the algorithm repeatedly iterates over the program and checks whether
the assignments contain common symbols with $I$ (lines 8--12).  If yes, all
the symbols in this assignment are also added to $I$, and the edge is
recorded in~$E$.  This is repeated until a fixpoint is reached (i.e., $I$
has not changed in an iteration).  In addition to assignments, symbols might
also be propagated via functions calls and thread creation (from arguments
to parameters of the function/thread, and from the return expression to the
left-hand side of the call or the argument of the join).  This is handled in
lines 13--21.

After we have gathered all affecting assignments, the second phase of the
algorithm begins (Algorithm~\ref{alg:dependency_analysis}, lines 2--8). This
phase additionally identifies the function calls that might affect the starting
expressions.
It is based on the following observation. If a function does not contain any
affecting assignments, and any of the functions it calls do not either (and
those that this function calls in turn do not, etc.), then the calls to the
function cannot affect the starting expressions.
The ability to prune function calls has a potentially big effect on the performance
of the analysis, as it can greatly reduce the amount of code that needs to be
analysed.

In the following section we evaluate the performance and effectiveness of the
dependency analysis. Its effect on the overall analysis is evaluated in Sec.~\ref{sec:exp}.

\paragraph{Evaluation}

We have evaluated the dependency analysis on a subset of 100 benchmarks of
the benchmarks given in Sec.~\ref{sec:exp}. For each benchmark the dependency
analysis was invoked with the set of starting expressions $\mathit{exprs}$ being
those occuring in lock operations or as the first argument of create and join
operations.
The results are given in the table below.

\vspace{-2ex}




\noindent
\begin{figure}[h!]
\scriptsize
\newcolumntype{Y}{>{\raggedleft\arraybackslash}X}
\begin{tabularx}{\linewidth}{p{2cm}YYY}
\toprule
& runtime & sign. assign. & sign. func.\\
\midrule
25th percentile & $0.03$\,s~~ &  0.3\%~~ & 43.3\%~~\\
arithmetic mean & $0.26$\,s~~ & 40.0\%~~ & 63.3\%~~\\
75th percentile & $0.38$\,s~~ & 72.2\%~~ & 86.5\%~~\\
\bottomrule
\end{tabularx}
\end{figure}

\noindent The table shows that the average time (over all benchmarks) to
perform the dependency analysis was $0.26$\,s.  The first and last line give
the 25th and 75th percentile.  This indicates for example that for $25\%$ of
the benchmarks it took $0.03$\,s or less to perform the dependency analysis. 
The third and fourth column evaluate the effectiveness of the analysis. 
On~average, $40\%$ of the assignments in a program were classified as
significant (i.e., potentially affecting the starting expressions).  The
data also shows that often the number of significant assignments was very
low (in $25\%$ of the cases it was $0.3\%$ or less).  This happens when the
lock usage patterns in the program are simple, such as using simple lock
expressions (like \texttt{pthread\_mutex\_lock(\&mutex)}) that refer to
global locks with simple initialisations (such as using static
initialization via \texttt{PTHREAD\_MUTEX\_INITIALIZER}).

The average number of functions classified as significant was $63.3\%$. This
means that on average $36.7\%$ of the functions that occur in a program were
identified as irrelevant by the dependency analysis and thus do not need to
be analysed by the following pointer analysis.
%
%

Overall, the data shows that the analysis is cheap and
able to prune a significant number of assignments and functions.

\begin{algorithm}[t]
\footnotesize
\DontPrintSemicolon
\SetKwInOut{Input}{Input}
\SetKwInOut{Output}{Output}
\SetKwProg{Fn}{function}{}{}

\Input{ICFA $\prog$, lock edges $\mathit{lock\_edges}$}
\Output{Set of affecting edges $A$}

$A \gets \mathit{affecting\_edges}(\prog, \mathit{lock\_edges})$\;
$F \gets \{ f~|~e \in A \wedge f = \getfunc(\src(e)) \}$\;
$F_h \gets \emptyset$\;
\While{$F \neq \emptyset$}{
  $\text{remove}~f~\text{from}~F$\;
  $F_h \gets F_h \cup \{f\}$\;
  $E \gets \{ e~|~ \getfunc(\tgt(e)) = f \wedge\newline
           \hspace*{10.8ex}\op(e) \in \{\funcentry, \threadentry\}\}$\;
  \For{$e \in E$}{
    $A \gets A \cup \{e\}$\;
    $f' \gets \getfunc(\src(e))$\;
    \If{$f' \notin F_h$}{
      $F \gets F \cup \{f'\}$\;
    }
  }
}
\Return{A}
\caption{Dependency analysis}
\label{alg:dependency_analysis}
\end{algorithm}

\begin{algorithm}[t]
\footnotesize
\DontPrintSemicolon
\SetKwInOut{Input}{Input}
\SetKwInOut{Output}{Output}
\SetKwProg{Fn}{function}{}{}

\Fn{affecting\_edges($\prog, \mathit{lock\_edges}$)}{
  $A \gets \mathit{lock\_edges}$\;
  $S \gets \bigcup_{e \in \mathit{lock\_edges}} \symbols(\op(e))$\;
  $R \gets \emptyset$\;
  \For{$e \in E(\prog)$}{
    $\mathit{op} \gets \op(e)$\;
    \If{\upshape $op = (a = b)\vee\newline
        \hspace*{2.8ex}op = \threadentry(\_, a, b) \vee\newline
        \hspace*{2.8ex}op = \funcexit(a, b) \vee\newline
        \hspace*{2.8ex}op = \threadjoin(a, b)$}{
      \vspace{2pt}
      $R \gets R \cup \{\symbols(a) \cup \symbols(b)\}$\;
    }
    \uElseIf{\upshape $op = \funcentry(\arg_1, \ldots, \arg_n,\newline
             \hspace*{24.8ex}\param_1, \ldots, \param_n)$}{
      \vspace{2pt}
      $R \gets R \cup \{\symbols(\arg_i) \cup \symbols(\param_i)~|~\newline
      \hspace*{11.3ex}i \in \{1, \ldots, n\}\}$\;
    }
  }
  $\mathit{NM} \gets \text{number\_map}(R)$\;
  $\mathit{SM} \gets \text{symbol\_map}(R)$\;
  $N_h, S_h \gets \emptyset, \emptyset$\;
  \While{$S \neq \emptyset$}{
    $\text{remove}~s~\text{from}~S$\;
    $S_h \gets S_h \cup \{s\}$\;
    \For{$n \in \mathit{SM}[s]$}{
      \If{$n \notin N_h$}{
        \vspace{1pt}
        $N_h \gets N_h \cup \{n\}$\;
        $S \gets S \cup (\mathit{NM}[r] - S_h)$\;
      }
    }
  }
  \For{$e \in E(\prog)$}{
    \If{\upshape $op = (a = b)\vee\newline
        \hspace*{2.8ex}op = \funcexit(a, b) \vee\newline
        \hspace*{2.8ex}op = \threadjoin(a, b)$}{
      \vspace{2pt}
      \If{\upshape $((\symbols(a) \cup \symbols(b)) \cap S_h) \neq \emptyset$}{
        \vspace{1pt}
        $A \gets A \cup \{e\}$\;
      }
    }
  }
  \Return{A}
}


\caption{Affecting edges}
\label{alg:affecting_assignments}
\end{algorithm}

\section{Non-Concurrency Analysis}\label{sec:nonconcurrent}

We have implemented an analysis (Algorithm~\ref{alg:nonconcurrent}) to compute
whether two places $\place_1, \place_2$ are non-concurrent. That is, the analysis determines whether the statements
associated with the places $\place_1, \place_2$ (i.e., the operations with which
the outgoing edges of $\placetop(\place_1), \placetop(\place_2)$ are labeled) cannot
execute concurrently in the contexts embodied by
$\place_1, \place_2$.

Whether the places are protected by a common lock is determined by computing the
intersection of the must locksets (lines 3--4). If the intersection is non-empty
they cannot execute concurrently and the algorithm returns $\true$. Otherwise
the algorithm proceeds to check whether the
places are non-concurrent due to create and join operations. This is done via a
graph search in the ICFA.
First the length of the longest common prefix
of $p_1$ and $p_2$ is determined (line 5). This is the starting point for the ICFA exploration. If
there is a path from $\loc_1$ to $\loc_2$, it is checked that all the threads
that are created to reach place $\place_1$ are joined before location
$\loc_2$ is reached (and same for a path from $\loc_2$ to $\loc_1$). This check
is performed by the procedure $\mathit{unwind()}$, the full details
of which we give in 
\rronly{Appendix~\ref{sec:non-concurrency-app}.}
\pponly{the extended version~\cite{kroening:2016}.}

We evaluated the non-concurrency analysis with respect to what fraction of all
the pairs of places $\place_1, \place_2$ of a program it classifies as
non-concurrent. We found that on a subset of 100 benchmarks of the
benchmarks of Sec.~\ref{sec:exp}, it classified $60\%$ of the places corresponding
to different threads as non-concurrent on average. We give more data in
\rronly{Appendix~\ref{sec:non-concurrency-app}.}
\pponly{the extended version~\cite{kroening:2016}.}
\begin{algorithm}
\footnotesize
\DontPrintSemicolon
\SetKwInOut{Input}{Input}
\SetKwInOut{Output}{Output}

\Input{places $p_1$, $p_2$, must locksets $\mustlockset^1$, $\mustlockset^2$}
\Output{$\true$ if $\place_1, \place_2$ are non-conc., $\false$ otherwise}

  \If{$p_1 = p_2$}{
    \Return{$\true$}
  }

  \If{$\mustlockset^1 \cap \mustlockset^2 \neq \emptyset$}{
    \Return{$\true$}
  }

  $i \gets |\commonprefix(\place_1, \place_2)|$
  
  $r_1, r_2 \gets \true, \true$\;
  $\loc_1, \loc_2 \gets \place_1[i], \place_2[i]$\;

  \If{\upshape \haspath($\loc_1$, $\loc_2$)}{
    $r_1 \gets unwind(i, \place_1, \loc_1, \loc_2)$
  }

  \If{\upshape \haspath($\loc_2$, $\loc_1$)}{
    $r_2 \gets unwind(i, \place_2, \loc_2, \loc_1)$
  }

  \Return{$r_1 \wedge r_2$}
\caption{Non-concurrency analysis}
\label{alg:nonconcurrent}
\end{algorithm}



\section{Lock Graph Analysis}\label{sec:lockgraph}

Our lock graph analysis consists of two phases. First, we build a lock
graph based on the lockset analysis. In the second phase, we prune
cycles that are infeasible due to information from the non-concurrency
analysis.

\subsection{Lock Graph Construction}\label{sec:buildlockgraph}

\begin{figure}
\footnotesize
\begin{tikzpicture}[nodestyle/.style={circle, draw, inner sep=1.5pt}]
\node[nodestyle] (m1) at (0,0) {$m_1$};
\node[nodestyle] (m2) at (2.35,0) {$m_2$};
\node[nodestyle] (m3) at (4.7,0) {$m_3$};
\node[nodestyle] (m4) at (5.5,0) {$m_4$};
\node[nodestyle] (m5) at (7.85,0) {$m_5$};
\draw[->] (m1) --node[auto]{(\texttt{m:\nolink{\ref{df:main:create}},t:\nolink{\ref{df:thread:lock2}}})} (m2);
\draw[->] (m2) --node[auto]{(\texttt{m:\nolink{\ref{df:main:create}},t:\nolink{\ref{df:thread:lock3}}})} (m3);
\draw[->] (m4) --node[auto]{(\texttt{m:\nolink{\ref{df:main:create}},t:\nolink{\ref{df:thread:lock5}}})} (m5);

\draw[->] (m3) edge[out=120, in=60] node[above]{(\texttt{m:\nolink{\ref{df:main:lock2}}})} (m2);

\path[->] (m1) edge[out=-20, in=-160] node[below]{(\texttt{m:\nolink{\ref{df:main:lock3}}})} (m3);

\path[->] (m5) edge[bend left] node[below]{(\texttt{m:\nolink{\ref{df:main:func2}},f:\nolink{\ref{df:func2:lock4}}})} (m4);
\end{tikzpicture}
\caption{\label{fig:example_lock_graph}
Lock graph for the program in Fig.~\ref{fig:example} (\texttt{t}, \texttt{m} and \texttt{f} are shorthand for \texttt{thread}, \texttt{main} and \texttt{func2}, respectively).
}
\end{figure}
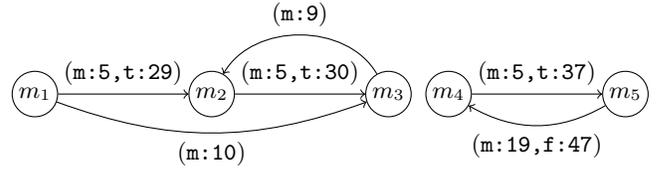





A \emph{lock graph} is a directed graph $L \in 2^{\objs^* \times
  \places \times \objs^*}$ (with $\objs^* = \objs \cup
\{\unknownlock\}$). Each node is a $\lock \in \objs^*$, and an
edge $(\lock_1, \place, \lock_2) \in \objs^* \times
  \places \times \objs^*$ from $\lock_1$ to $\lock_2$ is
labelled with the place $\place$ of the lock operation that acquired $\lock_2$ while $\lock_1$
was owned by the same thread $\getthread(\place)$.  Hence, the
directed edges indicate the order of lock acquisition.
Fig.~\ref{fig:example_lock_graph} gives the lock graph
for the example program in Fig.~\ref{fig:example}.

We use the result of the may lockset analysis
(Sec.~\ref{sec:lockset}) to build the lock graph.
Fig.~\ref{fig:build_lock_graph} gives the lock graph domain
that is instantiated in our analysis framework.
For each lock operation in place $\place$ a thread may acquire a lock
$\lock_2$ corresponding to the value set of the argument to the lock
operation.  This happens while the thread may own any lock $\lock_1$
in the lockset at that place.  Therefore we add an edge
$(\lock_1,\place,\lock_2)$ for each pair $(\lock_1,\lock_2)$.

Finally, we have to handle the indeterminate locks, denoted by
$\unknownlock$.  We~compute the closure $cl(L)$ of the graph w.r.t.~edges
that involve~$\unknownlock$ by adding edges from all predecessors of the
$\unknownlock$~node to all other nodes, and to each successor node of the
$\unknownlock$~node, we add edges from all other nodes.

\begin{figure}[t]
\footnotesize

\hrulefill

\vspace{0.5ex}

Domain: $2^{\objs^* \times \places \times \objs^*}$

\hrulefill

$s_1 \join s_2 = s_1 \cup s_2$

\hrulefill



With $\op(e) = \lockop(a)$:

\vspace{0.5ex}

$
\begin{array}{r@{\,}l}
\transf{e, \place}(s) =
s \cup \{(\lock_1,\place,\lock_2) \mid & \lock_1 \in \maylockset(\place), \\
&  \lock_2\in vs(\place,a) \}  
\end{array}
$

\vspace{0.5ex}

\hrulefill

\vspace{0.5ex}

$
\begin{array}{@{}r@{\,}l@{\,}l}
 \mathit{cl}(s) = s  & \cup \{(\lock_1,\place,\lock) \mid &
(\lock_1,\place,\unknownlock) \in s, \\ && \lock \in
\textsf{get\_locks}(s)\setminus\{\lock_1,\unknownlock\} \}\\ 
& \cup \{(\lock,\place,\lock_2) \mid &
(\unknownlock,\place,\lock_2) \in s, \\ && \lock \in
\textsf{get\_locks}(s)\setminus\{\lock_2,\unknownlock\} \}
\end{array}
$

\vspace{0.5ex}
\hrulefill

\vspace{-1ex}

\caption{\label{fig:build_lock_graph} Lock graph construction}
\end{figure}

\subsection{Checking Cycles in the Lock Graph}\label{sec:cyclecheck}

The final step is to check the cycles in the lock graph.
For this purpose we use the information from the non-concurrency analysis.
Each cycle $c$ in the lock graph could be a potential deadlock.
A cycle $c$ is a set of (distinct) edges; there is a finite number of
such sets.
A cycle is a potential deadlock if 
$
|c|>1 \wedge \textsf{all\_concurrent}(c) 
$
where
%
%
\begin{eqnarray*}
\lefteqn{\textsf{all\_concurrent}(c) \Leftrightarrow }\\
 && \forall (\lock_1,\place,\lock_2),(\lock'_1,\place',\lock'_2) \in c: \\ 
 && \neg \textsf{non\_concurrent}(\place,\place') \vee\\
 && (\textsf{get\_thread}(p) = \textsf{get\_thread}(p') \wedge\\
 && \textsf{multiple\_thread}(\getthread(\place)))
\end{eqnarray*}
%
and $\textsf{multiple\_thread}(t)$ means that $t$ was created
  in a loop or recursion.
%
Due to the use of our non-concurrency analysis we do not require any special
treatment for gate locks or thread segments as in~\cite{ABF+10}.

\section{Experiments}\label{sec:exp}

We implemented our deadlock analyser as a pipeline of static analyses in the
CPROVER framework,%
\footnote{Based on r6268 of \url{http://www.cprover.org/svn/cbmc/trunk}}
and we performed experiments to support the following hypothesis:
\textit{Our analysis handles real-world C code in a precise and efficient~way.}
We used 715 concurrent C programs that contain locks from the Debian GNU/Linux distribution,
with the characteristics shown in Fig.~\ref{tab:statistics}.%
\footnote{Lines of code were measured using \texttt{cloc} 1.53.
}
%
%
%
The table shows that the minimum number of different locks and lock operations
encountered by our analysis was $0$. We found that this is due to a small number
of benchmarks on which the lock operations were not reachable from the main
function of the program (i.e., they were contained in dead code).

We additionally selected 8 programs and introduced
deadlocks in them.
This gives us a benchmark set consisting of 723 benchmarks with 
a total of 7.1\,MLOC. Of these, 715 benchmarks are assumed to
be deadlock-free, and 8 benchmarks are known to have deadlocks.
The experiments were run on a
Xeon X5667 at 3\,GHz
running Fedora 20 with 64-bit binaries.
Memory and CPU time were restricted to 24\,GB and 1800 seconds per
benchmark, respectively.


\vspace*{-1.5ex}
\paragraph{Results}
We correctly report (potential) deadlocks in the 8 benchmarks with
known deadlocks.  The results for the deadlock-free programs are shown
in the following table grouped by benchmark size (t/o\ldots timed out,
m/o\ldots out of memory):


\smallskip
\noindent\hspace*{2ex}\begin{tabular}{c|rrrrrr}
KLOC & analysed & \bf proved & alarms & t/o & m/o \\
\hline
\hphantom{0}0--5\hphantom{0} & 252 & \bf 130 & 55 & 48 & 19 \\
\hphantom{0}5--10 & 269 & \bf 91 & 35 & 108 & 35 \\
10--15 & 86 & \bf 16 & 9 & 54 & 7 \\
15--20 & 93 & \bf 23 & 5 & 64 & 1 \\
20--50 & 15 & \bf 2 & 1 & 11 & 1 \\
\end{tabular}

\smallskip\noindent
For 105 deadlock-free benchmarks, we report alarms
that are most likely spurious. The main reason for such false
alarms is the imprecision of the pointer analysis with respect to
dynamically allocated data structures. This leads to lock operations
on indeterminate locks (see statistics in Fig.~\ref{tab:statistics}).
This is a challenging issue to solve, as we discuss in
Sec.~\ref{sec:concl}.

Scatter plots in Figs.~\ref{fig:locs_time_dead} and~\ref{fig:locs_mem_dead}
illustrate how the tool scales in terms of running time and memory
consumption with respect to the number of lines of code.  The tool
successfully analyzed programs with up to 40K lines of code.  As~the plots
show, the asymptotic behaviour of the algorithms in terms of lines of code
is difficult to predict since it mostly depends on the complexity of the
pointer analysis.

We evaluated the impact of the different analysis features on a random
selection of 83 benchmarks and break down the running times into the
different analysis phases on those benchmarks where the tool does not
time out or goes out of memory:
We found that the dependency analysis is effective at decreasing both the
memory consumption and the runtime of the pointer analysis. It decreased the
memory consumption by 27\% and the runtime by 60\% on average.
%
%
We observed that still the vast majority of the
running time (93\%) of our tool is spent in the pointer analysis, which is due to
the often large number of general memory objects, including all heap
and stack objects that may contain locks. May lock analysis (3\%),
must lock analyisis (2\%), and lock graph construction (2\%) take less
time; the run times for the dependency analysis and the cycle checking
(up to the first potential deadlock) are negligible.  For 80.2\% of
lock operations the must lock analysis was precise.

\vspace*{-1.5ex}
\paragraph{Comparison with other tools}
We tried to find other tools to experimentally compare with ours. However,
we did not find a tool that handles C code and with which a reasonable
comparison could be made.  Other tools are either for Java (such as
\cite{NPSG09}), are based solely on testing (Helgrind~\cite{helgrind}), or
are semi-automatic lightweight approaches relying on user-supplied
annotations (LockLint~\cite{locklint}).


\begin{figure}[t]
\begin{subfigure}{0.45\textwidth}
\centering
\begin{tikzpicture}[scale=0.7]
\begin{axis} [ymode=log]
\addplot [mark size=1,only marks] table {locs_time_no_deadlocks.csv};
\end{axis}
\end{tikzpicture}
\caption{LOC vs.~user time (timeout 1800\,s)\label{fig:locs_time_dead}
\vspace{0.2cm}} 
\end{subfigure}
\begin{subfigure}{0.45\textwidth}
\centering
\begin{tikzpicture}[scale=0.7]
\begin{axis} [ymode=log]
\addplot [mark size=1,only marks] table {locs_mem_no_deadlocks.csv};
\end{axis}
\end{tikzpicture}
\caption{\label{fig:locs_mem_dead}%
LOC vs.~memory consumption (memory limit 24\,GB)}
%
%
\end{subfigure}
\vspace{-1ex}
\caption{Experimental results}
\end{figure}
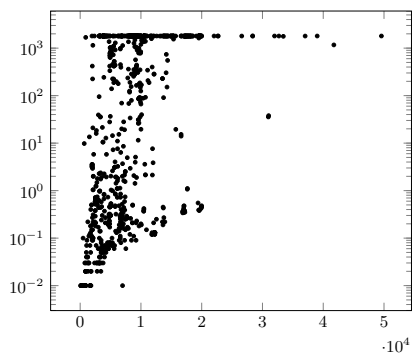
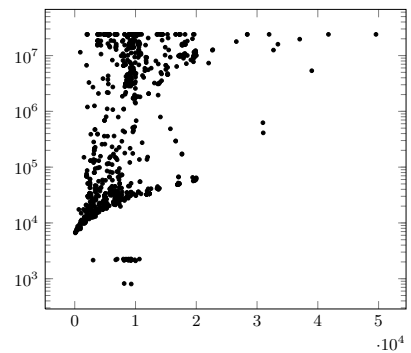


\begin{figure*}[t]
\scriptsize
\centering

\begin{subfigure}{0.5\textwidth}
\centering
    \begin{tabular}{|l||r|r|r|}
\hline
      & max & avg & min \\
    \hline
    \# lines of code	                    & 41,749 & 8,376.5 & 86\\
    \# Threads                   & 163  & 3.8 & 1 \\
    \# Threads in loop & 162 & 2.1 & 0\\
    \# Locks                     & 16 & 1.5 & 0\\
    \# Lock operations           & 30773 & 331.6 & 0\\
\hline
Precise must analysis                & 100\% & 80.2\% & 0\%  \\
Size of largest lockset                     & 8.0 & 1.2 & 1.0  \\
\# indeterminate locking operations   & 2106.0 & 32.3 & 0.0 \\
\# non-concurrency checks      & 450.0 & 1.6 & 0.0 \\
    \hline
    \end{tabular}
\end{subfigure}%
\begin{subfigure}{0.5\textwidth}
\centering
\begin{tabular}{|l|rrr|}
\hline
 & max & avg & min   \\
\hline
Total analysis time (s)                        & 1291.3 & 435.3 & 0.0\\
\hspace*{1em}Dependency analysis                & 5.0 & 0.1 & 0.0  \\
\hspace*{1em}Pointer analysis                & 1185.1 & 419.0 & 0.0  \\
\hspace*{1em}May lockset analysis           & 345.4 & 12.5 & 0.0 \\
\hspace*{1em}Must lockset analysis          & 29.9 & 1.1 & 0.0 \\
\hspace*{1em}Lock graph construction         & 31.3 & 1.2 & 0.0  \\
\hspace*{1em}Cycles detection                 & 327.4 & 1.4 & 0.0 \\
\hline
Peak memory (GB)                             & 24.0 & 6.3 & 0.008  \\
\hline
\end{tabular}
\end{subfigure}
\vspace{-1ex}
\caption{\label{tab:statistics}
Benchmark characteristics and analysis statistics 
(for the 367 benchmarks with no time out or out of memory)}
\end{figure*}

\vspace*{-1.5ex}
\section{Threats to Validity}
\label{sec:threats}


This section discusses the threats to internal and external validity of our
results and the strategies we have employed to mitigate them~\cite{feldt:2010}.

The main threat to internal validity concerns the correctness of our
implementation. To mitigate this threat we have continually tested our tool 
during the implementation phase, which has resulted in a testsuite of 122
system-, unit-, and regression-tests. To further test the soundness claim of our
tool on larger programs, we introduced deadlocks into 8 previously deadlock-free
programs, and checked that our tool correctly detected them. While we
have reused existing locks and lock operations to create those deadlocks, they
might nevertheless not correspond well to deadlocks inadvertently introduced by
programmers.

The threats to external validity concern the generalisability of our results
to other benchmarks and programming languages.  Our benchmarks have been
drawn from a collection of open-source C programs from the Debian GNU/Linux
distribution~\cite{kroening:2014} that use pthreads and contain lock
operations, from which we ran most of the smaller ones and some larger ones. 
We found that the benchmark set contains a diverse set of programs. 
However, we did not evaluate our tool on embedded (safety- or
mission-critical) software, a field that we see as a prime application area
of a sound tool like ours, due to the ability to verify the absence of
errors.

During the experiments we have used a timeout of 1800\,s. This in particular
means that we cannot say how the tool would fare for users willing to invest
more time and computing power; in particular, the false positive rate could
increase as programs that take a long time to analyse are likely larger and
could have more potential for deadlocks to occur.

Finally, our results might not generalise to other programming
languages. For example, while Naik et al.~\cite{NPSG09} (analysing
Java) found that the ability to detected common locks was crucial to
lower the false positive rate, we found that it had little effect in
our setting, since most programs do not acquire more than 2 locks in a
nested manner (see Tab.~\ref{tab:statistics}, size of largest
lockset).




\vspace{-1.5ex}
\section{Related Work}
\label{sec:related}

Deadlock analysis is an active research area.  Since
the mid-1990s numerous tools, mainly for C and Java, have been
developed. One can distinguish dynamic and static approaches.  A
common deficiency is that all these tools are neither sound nor
complete, and produce false positives and negatives. 


\vspace*{-1.3ex}
\paragraph{Dynamic tools}
The development of the Java PathFinder tool~\cite{Hav00,ABF+10} led to
ground-breaking work over more than a decade to find lock acquisition
hierarchy violation with the help of lock graphs, exposing the issue
of gatelocks, and segmentation techniques to handle different threads
running in parallel at different times.  \cite{AS06,JPSN09,Sor15} try
to predict deadlocks in executions similar to the observed one.
DeadlockFuzzer~\cite{JPSN09} use a fuzzying technique to search for
deadlocking executions.
Multicore SDK~\cite{LDQ11} tries to reduce the size of lock graphs by
clustering locks; and Magiclock~\cite{CC14} implements significant
improvements on the cycle detection algorithms.
Helgrind~\cite{helgrind}
is a popular open source dynamic deadlock detection tool, 
and there are many commercial implementations of dynamic deadlock
detection algorithms.

\vspace*{-1.3ex}
\paragraph{Static tools} 
There are very few static deadlock analysis tools for C.
LockLint~\cite{locklint}
relies on a user-supplied lock acquisition order, and is thus not
fully automatic.
RacerX~\cite{EA03} focuses more on fast analysis of large code bases
than soundness. It performs a path- and context-sensitive analysis,
but its pointer analysis is very rudimentary.
For Java there is Jlint2~\cite{AB01}, a tool similar to LockLint.
The tool Jade~\cite{NPSG09} consciously uses a may analysis 
instead of a must analysis, which causes unsoundness.
The tools presented in \cite{WTE05} and
\cite{Pra04} do not consider gatelock scenarios, which leads
to false alarms.

\vspace*{-1.3ex}
\paragraph{Other tools} 
Some tools combine dynamic approaches and constraint solving. For
example, CheckMate~\cite{JNSG10} model-checks a path along an observed
execution of a multi-threaded Java program; Sherlock~\cite{EP14} uses
concolic testing; and \cite{Sor15,AS06} monitor runtime executions of
Java programs. There are related techniques to detect synchronisation
defects due to blocking communication, e.g.~in message passing (MPI)
programs~\cite{FKNS14,CLC+12}, for the modelling languages BIP
(DFinder tool~\cite{BBNS09,BGL+11}) or ABS (DECO tool~\cite{FAG13})
that use similar techniques based on lock graphs and
may-happen-in-parallel information.

\vspace*{-1.3ex}
\paragraph{Dependency analysis}

Our dependency analysis is related to work on concurrent program
slicing~\cite{krinke:1998,krinke:2003,nanda:2000} and alias
analysis~\cite{burke:1994,kahlon:2007}.  Our analysis is more lightweight
than existing approaches as it works on the level of variable identifiers
only, as opposed to more complex objects such as program dependence graphs
(PDG) or representations of possible memory layouts.  Moreover, our analysis
disregards expressions occuring in control flow statements (such as
if-statements) as these are not relevant to the following pointer analysis
which consumes the result of the dependency analysis.  The analysis thus
does not produce an \emph{executable subset} of the program statements as in
the original definition of slicing by Weiser~\cite{weiser:1981}.

\vspace*{-1.3ex}
\paragraph{Non-concurrency analysis} 

Our non-concurrency analysis is context-sensitive, works on-demand, and can
classify places as non-concurrent based on locksets or create/join. 
Locksets have been used in a similar way in static data race
detection~\cite{EA03}, and Havelund~\cite{ABF+10} used locksets in dynamic
deadlock detection to identify non-concurrent lock statements.  Our handling
of create/join is most closely related to the work of Albert et
al.~\cite{albert:2012}.  They consider a language with asynchronous method
calls and an await statement that allows to wait for the completion of a
previous call.  Their analysis works in two phases, the second of which can
be performed on-demand, and also provides a form of context-sensitivity. 
Other approaches, which however do not work on-demand, include the work of
Masticola and Ryder~\cite{masticola:1993} and Naumovich et
al.~\cite{naumovich:1998} for ADA, and the work of Lee et
al.~\cite{lee:2012} for async-finish parallelism.

\vspace*{-1.2ex}
\section{Conclusions}\label{sec:concl}

We presented a new static deadlock analysis approach for concurrent C/pthreads
programs.
We demonstrated that our tool can effectively prove deadlock freedom
of 2.6\,MLOC concurrent C code from the Debian GNU/Linux distribution.
Our experiments show that the pointer analysis is the crucial
component of a static deadlock analyser. It takes most of the time, by
far, and is the primary source for false alarms when the arguments of
lock operations cannot be determined.
In our evaluation, we observed that the limitations of our pointer analysis
regarding dynamically allocated data structures, e.g.~when lock
objects are stored in lists, are the key reason 
for false alarms. Future work will focus on addressing this limitation.
Moreover, we will integrate the analysis of pthread
synchronisation primitives other than mutexes, e.g.~condition
variables, and extend our algorithm to Java synchronisation constructs.
We also want to go beyond lock hierarchy violations,
and will include the information from a termination analysis to 
detect loop-related deadlocks.
All this is part of a larger endeavour to show synchronisation correctness
and deadlock freedom of
all concurrent C programs that use locks (currently 3748
programs with 264.5\,MLOC) in the Debian GNU/Linux distribution.

\section{Acknowledgments}

This work is supported by ERC project 280053
and SRC task 2269.002.

\bibliographystyle{abbrv}
\bibliography{biblio}


\newpage

\begin{theappendix}

\appendix

\section{Flow-Insensitive Framework}
\label{sec:framework_fi}

In the flow-insensitive framework, the topmost location of a place always
corresponds to the entry point of a function. That is, we associate a sound
overapproximation of the data flow facts that hold for all locations in the
function with the entry point of the function. Fig.~\ref{fig:framework_fi}
gives the formalization of the context- and thread-sensitive flow-insensitive
framework. It reuses many definitions from the flow-sensitive formalization of
Fig.~\ref{fig:framework}.

The result of the flow-insensitive analysis is
defined as the least fixpoint of the following equation:


\hspace{-4ex}
\begin{tikzpicture}
\node[right] (a) at (0, 0) {
\begin{minipage}{\linewidth}
\begin{flalign*}
s = s_0~\join~&\lambda\,p. \transfi{\getfunc(p), p}(s(p))~\join_c\\
              &\bigsqcup_{p', e~\text{s.t.} \textsf{np}(p, p', e)} \transfi{e, p'}(s(p'))
\end{flalign*}
\end{minipage}
};
\node (c) at (4.26, -0.45) {\scriptsize$c$};
\end{tikzpicture}

\vspace{-3ex}

\begin{flalign*}
\text{with}~~\textsf{np}(p, p', (\loc_1, \loc_2))~=
                              ~&\getfunc(\loc_1) \neq \getfunc(\loc_2)\\
                              ~&\entryloc(\loc_1) = \textsf{top}(p')\,\wedge\\
                              ~&\entryloc(\loc_2) = \textsf{top}(p)\,\wedge\\
                              ~&\placenexti((\loc_1, \loc_2), p') = p\\[1ex]
\text{with}~~s \join s'~=~&\lambda\, p.\, s(p) \join_c s'(p)
\end{flalign*}


\begin{figure}[t]
\normalsize

\hrulefill
\vspace{1ex}

Domain: $\domi = \fpms \times \doma$

\vspace{-0.5ex}
\hrulefill
\vspace{-3ex}


\begin{align*}
\statei^{1} \joini \statei^{2} = & \statei^{1} \joins \statei^{2}
\end{align*}

\vspace{-2ex}


\hrulefill
\vspace{1ex}




With $\cfaedge = (\loc_1, \loc_2)$, $\placetop(\place) = \entryloc(\loc_1)$, $f = \getfunc(\loc_2)$, and $n = |p|$:\\

$\placenexti(\cfaedge, \place) =$
\[
\begin{cases}
  \entry_s(\place, \loc_2)         &\op(\cfaedge) \in \{\threadentry, \funcentry\}\\
  \place[:n-2] + \entryloc(\loc_2) &\op(\cfaedge) \in \{\funcexit, \threadexit,\\
  &\hspace{9.7ex}\threadjoin\}\\
  \place[:n-1]                     &\text{otherwise}
\end{cases}
\]






\vspace{-0.5ex}
\hrulefill
\vspace{-1.5ex}

\begin{tikzpicture}
\node[right] (a) at (0, 0) {
\begin{minipage}{\linewidth}
\begin{flalign*}
\transfi{f, p} =~~~~~\bigcirc~~~~~\transfi{e, p}
\end{flalign*}
\end{minipage}
};
\node (c) at (4.55, -0.52) {\scriptsize$e\,s.t.\,P(f, e)$};
\end{tikzpicture}

\vspace{-1.5ex}
\hrulefill
\vspace{1ex}


$
\transfi{\cfaedge, \place}(\statei) = \transfs{\cfaedge, \place}(\statei)
$

\vspace{0ex}
\hrulefill

\caption{\label{fig:framework_fi}
Context-, thread-, and flow-insensitive framework}
\end{figure}

\section{Non-Concurrency Analysis}
\label{sec:non-concurrency-app}

We describe how the analysis determines whether two places $\place_1, \place_2$
are non-concurrent due to the relationship between threads arising from create
and join operations.

The analysis is based on performing a graph search in the ICFA. It makes use
of three basic functions on directed graphs.  The function $\haspath(\loc_1,
\loc_2)$ returns true when there is a path in the ICFA between locations
$l_1$ and $l_2$.  The function $\onallpaths(\loc_1, \loc_2, \loc_3)$ returns
true when all paths in the ICFA from $\loc_1$ to $\loc_3$ pass through
$\loc_2$.  It is implemented by computing the set of dominators of $\loc_3$
(assuming $\loc_1$ to be the entry point) and then checking whether $\loc_2$
is contained in that set.  The function $\inloop(l)$ returns true when there
is a path in the ICFA that starts and ends in~$\loc$.

\paragraph{Algorithm}

We explain the algorithm on an example (Fig.~\ref{fig:nc_example}). The
example consists of four threads (including the main thread). We want to
determine whether the statements \texttt{x=1} and \texttt{x=2} are
non-concurrent. We see that they cannot run concurrently as
\texttt{main()} joins with \texttt{thread1()} before starting
\texttt{thread3()} and \texttt{thread1()} joins with \texttt{thread2()}
before returning.

Let us now look at how our algorithm establishes this fact. The algorithm is
called with places $p_1 = (\nolink{\ref{nc:main:create1}},
\nolink{\ref{nc:thread1:create}}, \nolink{\ref{nc:thread2:x1}})$ and $p_2 =
(\nolink{\ref{nc:main:create2}}, \nolink{\ref{nc:thread3:x2}})$.  We have no
locks in the example and hence the must locksets are empty (line 3).  Line 5
determines the length of the longest common prefix of $p_1$ and $p_2$ (which
is 0 in this case).  This is the starting point for the exploration.

The algorithm then checks whether there is a path from
\nolink{\ref{nc:main:create1}} to \nolink{\ref{nc:main:create2}} (line 8). 
If there would not be a path from either \nolink{\ref{nc:main:create1}} to
\nolink{\ref{nc:main:create2}} or \nolink{\ref{nc:main:create2}} to
\nolink{\ref{nc:main:create1}} this would mean that
\nolink{\ref{nc:main:create1}} and \nolink{\ref{nc:main:create2}} occur in
conflicting branches (e.g., one in the then- and the other in the
else-branch of an if statement) and thus the places could not be concurrent. 
In the current case there is a path from \nolink{\ref{nc:main:create1}} to
\nolink{\ref{nc:main:create2}}.

The algorithm then invokes $\mathit{unwind()}$, which checks that the
threads that are created to reach place $\place_1$ are all joined before
location \nolink{\ref{nc:main:create2}} is reached.  It does so by iterating
over $\place$ starting from the end.  The operation $\placetop(\place)$
returns the last element of $\place$, and the operation $\pop(\place)$
returns $\place$ with the last element removed.

The variable $\mathit{joined}$ indicates whether the last thread that was
created has been joined yet.  If the top element of $\place$ corresponds to
a create operation (line 10), then if $\mathit{joined}$ is $\false$ the
function returns $\false$.  If not, then $\mathit{joined}$ is set to
$\false$, and the place corresponding to the create is recorded in $p_c$. 
Then, a matching join for the create is searched (line 16) by invoking the
$\mathit{find()}$ function.

The function $\mathit{find}(\place_c, \place, \loc_1, \loc_2)$ takes the
place~$\place_c$, a place~$\place$, and locations $\loc_1$ and $\loc_2$. 
The locations $\loc_1$ and $\loc_2$ are in the same function (let $f$ =
$\getfunc(l_1)$), and the place $p$ has as top element the call to the
function $f$.  The function $\mathit{find()}$ looks for a matching join to
the create at place~$p_c$.  It does so by looking in the function $f$ (lines
3--6), and (recursively) in the callees of $f$ (lines 7--14).  The join must
occur on all paths between $\loc_1$ and $\loc_2$ (lines 5, 9).  The call
$\mathit{match}(\place_c, \place + \loc_{join})$ checks that the join at
place $\place + \loc_{join}$ matches the create at~$\place_c$ (i.e., the
thread ID returned by the \texttt{pthread\_create()} is the same as the one
passed to the \texttt{pthread\_join()}).


If the creation site is in a loop, we additionally ensure
that the thread is joined also on each path that goes back to the same location
(lines 17--20). The final lines 21--31 are like the loop body and handle the
locations $\loc_1$ and $\loc_2$.

For our example, for the first iteration of the while loop in line 4 we have
$p = (\nolink{\ref{nc:main:create1}}, \nolink{\ref{nc:thread1:create}}, \nolink{\ref{nc:thread2:x1}})$.
In this case, $\nolink{\ref{nc:thread2:x1}} \notin \mathit{create\_locs} \wedge \mathit{joind}$ in line 7 and we thus continue with the next iteration. 
Now we have $p = (\nolink{\ref{nc:main:create1}}, \nolink{\ref{nc:thread1:create}})$ and $\nolink{\ref{nc:thread1:create}} \in \mathit{create\_locs}$ and
thus set $\mathit{joined}$ to $\false$ and record $p_c = (\nolink{\ref{nc:main:create1}}, \nolink{\ref{nc:thread1:create}})$.
The invocation of $\mathit{find()}$ (line
16) finds the join in line \nolink{\ref{nc:thread1:join}}, and thus $\mathit{joined}$ is set to $\true$. The while loop
then terminates as $|(\nolink{\ref{nc:main:create1}})| \le i+1$. Lines 21--31 then look for a matching
join for the create at location \nolink{\ref{nc:main:create1}}. Again such a join is found and $\mathit{unwind()}$
returns $\true$. The algorithm thus overall returns $\true$.

\paragraph{Evaluation}

We have evaluated the non-concurrency analysis on a subset of 100 benchmarks of
the benchmarks described
in Sec.~\ref{sec:exp}. For each benchmark we randomly selected 1000 pairs of
places $(\place_1, \place_2)$ such that $\place_1$ and $\place_2$ correspond to
different threads. We then performed the non-concurrency check for each of the
1000 pairs. The results are given in the table below.

\vspace{2ex}





\noindent
\begin{minipage}{\linewidth}
\scriptsize
\newcolumntype{Y}{>{\raggedleft\arraybackslash}X}
\begin{tabularx}{\linewidth}{p{2cm}YYp{2cm}}
\toprule
& runtime & n.c.~places & n.c.~lock places\\
\midrule
25th percentile & $0.14$s~~ & 47\%~~ & \hspace{11ex}11\%\\
arithmetic mean            & $4.07$s~~ & 60\%~~ & \hspace{11ex}43\%\\
75th percentile & $4.56$s~~ & 79\%~~ & \hspace{11ex}67\%\\
\bottomrule
\end{tabularx}
\end{minipage}

\vspace{2ex}

\noindent
The table shows that the average
time (over all benchmarks) it took to perform 1000 non-concurrency checks was
$4.07$s. The first and last line give the 25th and 75th percentile. This
indicates for example that for $25\%$ of the benchmarks it took $0.14$s or less
to perform 1000 non-concurrency checks. The third and fourth column evaluate the
effectiveness of the non-concurrency analysis. The third column shows that on
average our analysis classified $60\%$ of the place pairs as non-concurrent. The
fourth column gives the same property while only regarding places that
correspond to lock operations. The number of places classified as non-concurrent
is lower in this case, which is expected as the code portions using locks are
those that can run concurrently with others.
Overall, the data shows that the non-concurrency
analysis is both fast and effective.

\begin{figure}[t]
\scriptsize
\begin{subfigure}{0.55\linewidth}
\begin{lstlisting}[countblanklines=false,mathescape]
int main()
{
  pthread_t tid1;
  pthread_t tid3;
  pthread_create(|\label{nc:main:create1}|
    &tid1, 0, thread1, 0);
  pthread_join(tid1);|\label{nc:main:join}|
  pthread_create(|\label{nc:main:create2}|
    &tid3, 0, thread3, 0);
  return 0;
}
\end{lstlisting}
\end{subfigure}%
\begin{subfigure}{0.5\linewidth}
\begin{lstlisting}[mathescape,countblanklines=false,firstnumber=12]
void *thread1() {
  pthread_t tid2;
  pthread_create(|\label{nc:thread1:create}|
    &tid2, 0, thread2, 0);
  pthread_join(tid2, 0);|\label{nc:thread1:join}|
  return 0;
}

void *thread2() {
  x = 1;|\label{nc:thread2:x1}|
  return 0;
}

void *thread3() {
  x = 2;|\label{nc:thread3:x2}|
  return 0;
}
\end{lstlisting}
\end{subfigure}%
\caption{\label{fig:nc_example}
Statements \texttt{x=1} and \texttt{x=2} are non-concurrent}
\end{figure}

\begin{algorithm}[h]
\footnotesize
\DontPrintSemicolon
\SetKwInOut{Input}{Input}
\SetKwProg{Fn}{function}{}{}

\Fn{$\mathit{find}(\place_c, \place, \loc_1, \loc_2)$}{
$f \gets \getfunc(\loc_1)$\;
$\mathit{join\_locs} \gets \{\loc \in \locs(f)~|~\exists e = (\loc, \_)\colon \textsf{is\_join}(\op(e))\}$\;

\ForEach{\upshape $\loc_{join} \in join\_locs$}{
  \If{\upshape \onallpaths($\loc_1, \loc_{join}, \loc_2$) $\wedge$ \match($\place_c, \place + \loc_{join}$)}{
    \Return{$\true$}\;
  }
}

$\mathit{\mathit{entry\_edges}} \gets \{e = (\loc_{src}, \loc_{tgt}) ~|~ \getfunc(\loc_{src}) = f \wedge \op(e) = \funcentry\}$\; 


\ForEach{$(\loc_{src}, \loc_{tgt}) \in \mathit{entry\_edges}$}{
  \If{\upshape \onallpaths($\loc_1, \loc_{src}, \loc_2)$}{
    $p' \gets p + \loc_{src}$\;
    $f' \gets \getfunc(\loc_{tgt})$\;
    $r \gets find(\newline \hphantom{xxxxxx} \place_c, \newline \hphantom{xxxxxx} p', \newline \hphantom{xxxxxx} \loc_{tgt}, \newline \hphantom{xxxxxx} \exitloc(f'))$\;
    \If{$r$}{
      \Return{true}
    }
  }
}

\Return{false}

}
\caption{Find join}
\end{algorithm}

%
\begin{algorithm}
\footnotesize
\DontPrintSemicolon
\SetKwInOut{Input}{Input}
\SetKwProg{Fn}{function}{}{}

\Fn{$\mathit{unwind}(i, \place, \loc_1, \loc_2)$}{
$\mathit{create\_locs} \gets \{ \loc~|~ \exists e=(\loc, \_)\colon \op(e) = \threadentry\}$\;
$\mathit{joined} \gets true$\;
\While{\upshape $|p| > i+1$}{
$\loc \gets \placetop(\place)$\;
$f \gets \getfunc(\loc)$\;

  \If{\upshape $\loc \notin \mathit{create\_locs} \wedge \mathit{joined}$}{
    $p \gets \pop(\place)$\;
    \textbf{continue}\;
  }
  
  \If{\upshape $\loc \in \mathit{create\_locs}$}{
    \If{\upshape $\neg \mathit{joined}$}{
      \Return{false}
    }
    $\mathit{joined} \gets false$\;
    $\place_c \gets p$\;
  }

  $\place \gets \pop(\place)$\;
  $\mathit{joined} \gets \mathit{find}(\place_c, \place, \loc, \exitloc(f))$\;
  
  \If{\upshape \inloop($\loc$)}{
    $\mathit{loop\_joined} \gets \mathit{find}(\place_c, \place, \loc, \loc)$\;
    \If{$\neg \mathit{joined} \vee \neg \mathit{loop\_joined}$}{
      \Return{false}
    }
  }
}
  \If{\upshape $\loc_1 \in \mathit{create\_locs}$}{
    \If{\upshape $\neg \mathit{joined}$}{
      \Return{false}
    }
    $\mathit{joined} \gets false$\;
    $\place_c \gets \place$\;
  }

  \If{\upshape $\neg \mathit{joined}$}{
    $\place \gets \pop(\place)$\;
    $\mathit{joined} \gets \mathit{find}(\place_c, \place, \loc_1, \loc_2)$\;

    \If{\upshape $\inloop(\loc_1)$}{
      $\mathit{loop\_joined} \gets \mathit{find}(\place_c, \place, \loc_1, \loc_1)$\;
      \Return{$\mathit{joined} \wedge \mathit{loop\_joined}$}\;
    }
  }

  \Return{$\mathit{joined}$}
}
\caption{Unwind}
\end{algorithm}

\section{Must Lockset Analysis}
\label{sec:must_lockset_analysis_app}

Fig.~\ref{fig:must_lock_set_analysis} gives a formalisation of the must lockset
analysis. The must locksets can never contain the value $\unknownlock$.

\begin{figure}[t]
\footnotesize

\hrulefill

Domain: $2^{\objs} \cup \{\{\unknownlock\}\}$

\hrulefill

$s_1 \join s_2 = s_1 \cap s_2$

\hrulefill



With $\op(e) = \lockop(a)$:

$
\transf{e, p}(s) =
\begin{cases}
s \cup vs(p, a) & \text{if}~|vs(p, a)|=1 \wedge vs(p, a) \neq \{\unknownlock\}\\
s               & \text{otherwise}\\
\end{cases}
$

\vspace{2ex}

With $\op(e) = \unlockop(a)$:

$
\transf{e, p}(s) =
\begin{cases}
s - vs(p, a) & \text{if}~vs(p, a) \neq \{\unknownlock\}\\
\emptyset    & \text{otherwise}\\
\end{cases}
$

\vspace{2ex}

With $\op(e) \in \{\threadentry, \threadexit, \threadjoin\}$:

\vspace{1.5ex}

$
\transf{e, p}(s) = \emptyset
$

\vspace{1ex}
\hrulefill

\caption{\label{fig:must_lock_set_analysis}Must lockset analysis}
\end{figure}

\section{Framework Implementation}\label{sec:implementation}

Fig.~\ref{fig:trie_numbering} gives an example of a place map which contains the
mappings $(\loc_1) \leftrightarrow 1, (\loc_1, \loc_2) \leftrightarrow 0$, and $(\loc_3) \leftrightarrow 2$.
Unlike in an ordinary trie, in our implementation the nodes also have pointers
to their parent (dashed arrows). This allows to reconstruct a place from a
pointer to a node in the trie. For example, by starting from the leaf node
labeled with $0$ we can traverse the parent edges backwards to get the place $(\loc_1, \loc_2)$.

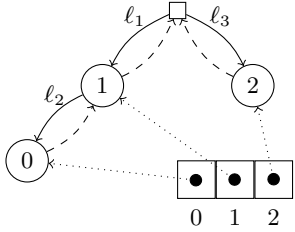
\begin{figure}

\centering

\begin{tikzpicture}

\node[rectangle, draw] (n1) at (0, 0) {};
\node[circle, draw] (n2) at (-1, -1) {1};
\node[circle, draw] (n3) at (1, -1) {2};
\node[circle, draw] (n4) at (-2, -2) {0};

\draw[->] (n1) edge[out=-135-20, in=45+20] node[above] {$\ell_1$} (n2);
\draw[->, dashed] (n2) edge[out=45-20, in=-135+20] (n1);

\draw[->] (n1) edge[out=-45+20, in=135-20] node[above] {$\ell_3$} (n3);
\draw[->, dashed] (n3) edge[out=135+20, in=-45-20] (n1);

\draw[->] (n2) edge[out=-135-20, in=45+20] node[above] {$\ell_2$} (n4);
\draw[->, dashed] (n4) edge[out=45-20, in=-135+20] (n2);

\begin{scope}[list/.style={draw, rectangle}, start chain, node distance=0pt,
inner sep=0.25cm, yshift=-2.25cm, xshift=0.25cm]

\node[list, on chain] (A) {};
\node[list, on chain] (B) {};
\node[list, on chain] (C) {};

\node[below of=A, yshift=-0.5cm] {0};
\node[below of=B, yshift=-0.5cm] {1};
\node[below of=C, yshift=-0.5cm] {2};

\node[inner sep=0, xshift=0.1cm, yshift=-0.02cm] (A1) {};
\node[inner sep=0, xshift=0.6cm, yshift=-0.065cm] (B1) {};
\node[inner sep=0, xshift=1.05cm, yshift=-0.1cm] (C1) {};

\draw[*->, dotted] (A1) -- (n4);
\draw[*->, dotted] (B1) -- (n2);
\draw[*->, dotted] (C1) -- (n3);

\end{scope}

\end{tikzpicture}

\vspace{-1ex}
\caption{Trie- and array-based place map data structure, representing the two way mapping $(\loc_1) \leftrightarrow 1$,
$(\loc_1, \loc_2) \leftrightarrow 0$, $\loc_3 \leftrightarrow 2$}
\label{fig:trie_numbering}
\end{figure}


\section{Correctness Proofs}\label{sec:proofs}

In this section, we show the correctness of our deadlock analysis approach. We
show that the may locksets computed for each place overapproximate the sets of
locks a thread may hold in a concrete execution at that place. Based on the
correctness of the may locksets, we then show that our lock graph
overapproximates the potential lock allocation graphs (LAGs). That is, we show
that for each execution prefix of the program for which there exists a cycle in
the LAG, there also is a cycle in the lock graph. We assume the
correctness of the pointer analysis and the correctness of the non-concurrency
analysis.

\subsection{Preliminaries}

In Sec.~\ref{sec:framework_details} we have formulated our analyses as a
fixpoint computation over the ICFA. An analysis computes data flow facts for
each \emph{place} (see the fixpoint equation in
Sec.~\ref{sec:framework_details}). The analysis can also be viewed in a slightly
different way: as computing a fixpoint over a larger structure that we term
\emph{context-sensitive control-flow automaton} (CCFA). The CCFA for a program
is just like the ICFA, but with the nodes being places rather than locations.
Two places $\place, \place'$ are connected by an edge if $\placetop(\place)$ and
$\placetop(\place')$ are connected by an edge in the corresponding ICFA.

We denote a concrete \emph{execution} (or execution prefix) of a program as an
interleaving $E = (\place_1, t_1, o_1)(\place_2, t_2, o_2)\ldots\allowbreak(\place_n, t_n,
o_n)$. The $\place_i$ are places, the $t_i$ are concrete thread IDs, and the
$o_i$ are execution instances of the operations with which the edges connecting
the places $\place_i, \place_{i+1}$ are labeled. We refer to the individual
tuples that make up $E$ as \emph{steps}.
We use
array subscript notation (0-based) to refer to individual elements of lists and
tuples. For example, $E[0][2]$ refers to the third component of the first tuple
of execution $E$. We further use slice notation to refer to contiguous
subsequences of executions. We further use slice notation (e.g., $E[n\colon\! m]$) for
contiguous subsequences of executions (see Sec.~\ref{sec:threadsensitive}).

We denote an \emph{execution path} (or an execution path prefix) of a program as $E_p =
(\place_1, \mathit{op}_1)(\place_2, \mathit{op}_2)\ldots(\place_n,
\mathit{op}_n)$. An execution path of a program is a path through its CCFA,
starting at the entry point. Here
the $\mathit{op}_i$ are the operations with which the edges connecting the
places are labeled (rather than concrete instances of those operations).

Given a thread ID~$t$ (resp.~abstract thread ID~$t'$), we denote by $E|t$
(resp.~$E_p|t'$) the executions (resp.~execution paths) projected to the
steps of the given thread $t$ (resp.~abstract thread $t'$).  We further
denote by $T(E) = \{ t~|~(\_, t, \_) \in E \}$ the set of threads in
execution $E$.\footnote{We use the same notation ($\in$) to denote elements
of sets and lists.} A thread always starts with a $\threadentry$ operation,
thus we have $(E|t)[0][2] = \threadentry(\ldots)$ and $(E_p|t')[0][1] =\allowbreak
\threadentry(\ldots)$.

\begin{property}\label{prop:basic}
Let $E$ be an execution prefix and let $E' = E|t$ ($n = |E'|$) be the execution
of some thread $t$ in $E$. Then, there is an execution path prefix $E_p$ ($m = |E_p|$)
with

\vspace{-3ex}

$$E_p[m-n\colon\!]|_{(p, \_)\mapsto p} = E'|_{(p, \_, \_)\mapsto p}$$
\end{property}







\noindent
The property holds by the shape of the CCFA (which directly derives from the
shape of the ICFA).
The property states that for each execution $E$ and thread $t$ in $E$, there is
an execution path prefix $E_p$ such that the sequence of places visited by $t$
and the sequence of places visited by the suffix of $E_p$ are the same.



We next define a concretisation function which states how the result of the
pointer analysis (and hence also the static locksets) are interpreted. 

\begin{definition}
Let $A$ be the set of all locks that may be held in any execution of a program.
Then:\\[-2pt]

$
c(\mathit{ls}) =
\begin{cases}
A           & \mathit{ls} = \{\star\}\\
\mathit{ls} & \text{otherwise}\\
\end{cases}
$

\end{definition}





\noindent
The function satisfies the property $c(\mathit{ls}_1) \cup c(\mathit{ls}_2) =
c(\mathit{ls}_1 \cup \mathit{ls}_2)$.
We are now in a position to prove the correctness of the may lockset analysis.

\subsection{May lockset correctness}

If $E = \allowbreak(\place_1, t_1, o_1)\allowbreak\ldots(\place_n, t_n, o_n)$ is an execution prefix
we denote by $\mathit{ls}_c(E)$ the concrete set of locks before executing the
final step $(\place_n, t_n, o_n)$. This is the set of locks held by thread $t_n$
at that step. A thread starts with an empty set of locks held.

\begin{theorem}
Let $\prog$ be a program and let $E = \allowbreak(\place_1, t_1, o_1)\allowbreak\ldots(p_n, t_n, o_n)$
be an execution prefix of $\prog$. Then:\\[-2pt]

\hspace*{4ex}$\mathit{ls}_c(E) \subseteq c(\mathit{ls}_a(p_n))$.
\end{theorem}

\vspace{0.1ex}

\begin{proof}
Let $t = t_n$ and let $E' = E|t$. By Property~\ref{prop:basic} there is an
execution path $E_p$ that ends in a sequence $E_p'$ which consists of the same
sequence of places as $E'$. 
We write $\lockset_c(i)$ for the concrete lockset before executing step $i$ of
$E'$. These locksets result from the execution of the program. We write
$\lockset_a(i)$ for the may lockset before handling step $i$ of $E_p'$. These
locksets are the result of applying the transfer function defined in
Fig.~\ref{fig:may_lock_set_analysis}.

We next show by induction that the may locksets computed by our analysis for
each step along $E_p'$ overapproximate the concrete locksets of $E'$ at the
corresponding steps. That is, we show that for all $0 \le i < |E'|\colon
\lockset_c(i) \subseteq c(\lockset_a(i))$.
Each thread starts with an empty lockset. In our analysis this is reflected by
the clause (2) in Fig.~\ref{fig:may_lock_set_analysis}. Hence the base case
$\lockset_c(0) = \emptyset \subseteq \emptyset = c(\emptyset) = c(\lockset_a(0))$ holds.

We next show the induction step via a case distinction based on whether step $i$
is (1) a lock operation or (2) an unlock operation.

\vspace{2ex}

\noindent
\textbf{(1)} Let $\lockset_c(i) \subseteq c(\lockset_a(i))$ and let $p = E'[i][0]$. We show that then also
$\lockset_c(i+1) \subseteq c(\lockset_a(i+1))$ after a lock operation $\lockop(a)$. We
perform a case distinction over the cases of the definition of
$\transf{.}(\lockset_a(i))$ (see Fig.~\ref{fig:may_lock_set_analysis}).

\vspace{1.2ex}

\begin{adjustwidth}{3mm}{}
\noindent
\textbf{(Case 1)} Let $l$ be the concrete lock acquired. By the correctness of the
        pointer analysis, we have $\{l\} \subseteq c(\mathit{vs}(p, a))$.
        Therefore:\\[-1.5ex]

        \noindent
        $\lockset_c(i+1) = \lockset_c(i) \cup \{l\} \subseteq c(\lockset_a(i)) \cup c(\mathit{vs}(p, a)) \subseteq c(\lockset_a(i) \cup \mathit{vs}(p, a)) = c(\lockset_a(i+1))$\\

\vspace{-1ex}

\noindent
\textbf{(Case 2)} $\lockset_c(i+1) \subseteq A = c(\{\unknownlock\})$ holds since $A$ is the set of all locks.
\end{adjustwidth}

\vspace{1ex}

\noindent
\textbf{(2)} Let $\lockset_c(i) \subseteq c(\lockset_a(i))$ and let $p = E'[i][0]$. We show that then also
$\lockset_c(i+1) \subseteq c(\lockset_a(i+1))$ after an unlock operation $\unlockop(a)$. We
perform a case distinction over the cases of the definition of
$\transf{.}(\lockset_a(i))$ (see Fig.~\ref{fig:may_lock_set_analysis}).

\vspace{1.5ex}

\begin{adjustwidth}{3mm}{}
\noindent
\textbf{(Case 1)} Since $|\lockset_a(i)| = 1$ and $\lockset_a(i) \neq \{\unknownlock\}$,
         we also have $|\lockset_c(i)| = 1$. Therefore, $\lockset_c(i+1) = \emptyset \subseteq \emptyset = c(\emptyset) = c(\lockset_a(i+1))$.

\vspace{1.5ex}

\noindent
\textbf{(Case 2)} Let $l$ be the concrete lock released. 
        Thus we have $l \in c(\lockset_a(i))$ and $l \in c(\mathit{vs}(p, a))$.
Since $|\lockset_a(i) \cap
        \mathit{vs}(p, a)| = 1$ we have that $c(\lockset_a(i) \cap \mathit{vs}(p, a)) = \{l\}$. Thus, $c(\lockset_a(i)
        - \mathit{vs}(p, a)) = c(\lockset_a(i)) - \{l\}$. Therefore, $\lockset_c(i+1) = \lockset_c(i) - \{l\} \subseteq
        c(\lockset_a(i)) - \{l\} = c(\lockset_a(i) - \mathit{vs}(p, a)) = c(\lockset_a(i+1))$.

\vspace{1.5ex}

\noindent
\textbf{(Case 3)} $\lockset_c(i+1) \subseteq \lockset_c(i)$ and thus $\lockset_c(i+1) \subseteq c(\lockset_a(i)) =
        c(\lockset_a(i+1))$

\end{adjustwidth}

\vspace{2ex}

Thus, we have shown that the may lockset is an overapproximation of the concrete
lockset at any step along $E_p'$, and
thus in particular also at the final step of $E_p'$ (i.e., at the place
associated with the final step of $E_p'$).
%
%
We can now use the properties of data flow analyses to complete the proof.
First, since the may lockset computed for the final place of $E_p$ is an
overapproximation of the concrete lockset, it follows that the ``meet over all
paths'' (MOP) at this place is an overapproximation of the concrete
locksets for all concrete executions that might reach that place.

Second, the analysis given in Fig.~\ref{fig:may_lock_set_analysis} consists of
a finite lattice with top element $\{\unknownlock\}$, join function $\join$, and
a monotonic transfer function.
%
%
Consequently, the minimal fixpoint solution (MFP) of the data flow
equations overapproximates the MOP solution. Therefore, since the MOP solution is
sound, the MFP solution is also sound.
\end{proof}

\subsection{Lock graph correctness}

During a concrete execution of a program, each step of the execution has an
associated \emph{lock allocation graph} (LAG). The LAG has two types of nodes: threads
and locks. There is an edge from a lock node to a thread node if the lock is
assigned to that thread (allocation edge). There is an edge from a thread node
to a lock node if the thread has requested the lock (a request edge). If the LAG
at a certain step in the execution has a cycle, then the involved threads have
deadlocked. Those threads cannot make any more steps (but other threads might).
We thus need to show that whenever there is an execution that has a cyclic LAG,
then our lock graph also has a cycle $c'$ for which $\textsf{all\_concurrent}(c')$
holds.

We first show a lemma about the lock graph closure computation (see
Fig.~\ref{fig:build_lock_graph}) that we will use later on.

\begin{lemma}\label{lem:closure}
Let $\lockset_1$, $\lockset_2$, $\lockset_3$, $\lockset_4$ be nonempty static
locksets, and let $\place, \place'$ be places. Let further $l \in c(\lockset_2)$
and $l \in c(\lockset_3)$. Let $L = \{ (\lock_1, \place'', \lock_2)~|~(\lock_1
\in \lockset_1 \wedge \lock_2 \in \lockset_2 \wedge p'' = p) \vee (\lock_1 \in
\lockset_3 \wedge \lock_2 \in \lockset_4 \wedge p'' = p') \}$. Then:

\vspace{-3ex}

\begin{multline*}
\forall \lock_1 \in \lockset_1, \lock_2 \in \lockset_4\colon \exists
\lock\colon\\
(\lock_1, \place, \lock), (\lock, \place', \lock_2) \in \mathit{cl}(L)
\end{multline*}
\end{lemma}

\begin{proof}~\\

\noindent
\textbf{(1)} Assume $\lock \in \lockset_2, \lock \in \lockset_3$ ($\lock$ may be $\unknownlock$). Then, by the definition
of $L$ above, for all locks $\lock_1 \in \lockset_1, \lock_2 \in \lockset_4, (\lock_1,
\place, \lock),\allowbreak (\lock, \place', \lock_2) \in L \subseteq \mathit{cl}(L)$.

\noindent
\textbf{(2)} Assume $\lockset_2 = \{\unknownlock\}, \lock \in \lockset_3, \lock \neq
\unknownlock$. Then for all locks $\lock_1 \in \lockset_1, \lock_2 \in
\lockset_4, (\lock_1, \place, \unknownlock), (\lock, \place', \lock_2) \in L$.
Then in $\mathit{cl}(L)$ there is an edge $(\lock_1, \place, \lock)$ by the
definition of $\mathit{cl}()$.

\noindent
\textbf{(3)} Assume $\lock \in \lockset_2, \lock \neq \unknownlock, \lockset_3 =
\{\unknownlock\}$. This case is symmetric to (2).
%
%
%
%
\end{proof}

\noindent
We next show a lemma about the definition of $\textsf{all\_concurrent}()$.

\begin{lemma}\label{lem:nonconcurrent}
Let $E$ be an execution prefix, let $G$ be the LAG at its final step, and let $c$ be a cycle in
$G$. Let $t_1, \ldots, t_n$ be the threads involved in the cycle $c$, and let
$(p_1, t_1, o_1),\allowbreak \ldots,\allowbreak (p_n, t_n, o_n)$ be last steps of each thread involved
in $c$ in $E$. Then for all $p_i$, $p_j$:

\vspace{-3.5ex}

\begin{multline*}
\neg \textsf{non\_concurrent}(p_i, p_j) \vee\\
(\textsf{get\_thread}(p_i) = \textsf{get\_thread}(p_j) \wedge \textsf{multiple\_thread}(p_i))
\end{multline*}

\end{lemma}

\begin{proof}
Since in $E$ all steps $(\place_1, t_1, o_1), \ldots, (\place_n), t_n, o_n)$
could reach a lock operation on which they blocked, they must have been able to
run concurrently in this execution. Now for two places $p_i, p_j,
\textsf{get\_thread}(p_i) \neq \textsf{get\_thread}(p_j)$, it follows that
$\neg \textsf{non\_concurrent}(p_i, p_j)$ by the correctness of the
non-concurrency analysis.

Now assume $\textsf{get\_thread}(p_i) = \textsf{get\_thread}(p_j)$. In this case
the places have the same abstract thread ID but they occur in \emph{different}
concrete threads. This occurs when a thread create operation occurs in a loop or
a recursion (or the call to the function that invokes the thread creation
operation occurs in a loop or recursion, etc.). Then we have
$\textsf{multiple\_thread}(p_i)$.
\end{proof}

We can now show the main theorem stating the soundness of our lock graph and
cycle search.

\begin{theorem}
Let $\prog$ be a program and let $E = (\place_1, t_1, o_1)\allowbreak\ldots(p_n,
t_n, o_n)$ be an execution prefix of $\prog$. Then if the LAG at the final step
of $E$ has a
cycle, then the lock graph of $\prog$ has a cycle $c'$ with
$\textsf{all\_concurrent}(c')$.
%
\end{theorem}

\begin{proof}
Let $c$ denote a cycle in the LAG at the final step of $E$.
%
%
In this cycle, every thread and every lock
occurs exactly once.
Let $t, t'$ be two threads involved in the cycle such that there are edges
$(l_1, t), (t, l_2), (l_2, t'), (t', l_3)$,
for locks $l_1, l_2, l_3$ (we might have that $l_1 = l_3$).

Let $n = |(E|t)|$, $m = |(E|t')|$, and let $(p, t, o) = (E|t)[n-1]$, $(p', t', o') = (E|t')[m-1]$ be the last steps of $E|t$, $E|t'$.
The steps $o, o'$ are lock operations, i.e., $o = \lockop(\mathit{exp}_1\colon l_2), o'
= \lockop(\mathit{exp}_2\colon l_4)$ with expression $\mathit{exp}_1$ referring to $l_2$ and
expression $\mathit{exp}_2$ referring to $l_4$.
That is, $l_2$ is the lock requested by $o$, and $l_4$ is the lock requested by
$o'$. Moreover, $l_1$ is in the lockset $\lockset_c$ at the last
step of $E|t$, and $l_2$ is in the lockset $\lockset_c'$ at the
last step of $E|t'$.
%

The analysis visits each place at least once, hence the transfer function (see
Fig.~\ref{fig:build_lock_graph}) is also applied to the edges outgoing from the
places $\place$ and $\place'$.
Applying the transfer function to the lock edge starting at $\place$ adds edges from
the elements of $\lockset_a(\place)$ to the elements of $\mathit{vs}(p, \mathit{exp}_1)$ to the lock graph $L$.
Applying the transfer function to the lock edge starting at $\place'$ adds edges from
the elements of $\lockset_a(\place')$ to the elements of $\mathit{vs}(p', \mathit{exp}_2)$ to the lock graph $L$.
By the correctness of the may lockset analysis and the correctness of the pointer
analysis we have $l_2 \in c(\mathit{vs}(\place, \mathit{exp}_1)), l_2 \in
c(\lockset_a(p')))$. Therefore, by Lemma~\ref{lem:closure}, it
follows that for all locks $\lock_1 \in \lockset_a(\place), \lock_2 \in
\mathit{vs}(\place', \mathit{exp}_2)$, there is a lock $\lock$ such that
$(\lock_1, \place, \lock),\allowbreak (\lock, \place', \lock_2) \in L$.

Thus, in the previous paragraph, we have shown that for any portion of the cycle $c$ consisting of adjacent
threads $t, t'$ and edges $(l_1, t), (t, l_2),\allowbreak (l_2, t'), (t', l_3)$, there is a portion
$(\lock_1, \place, \lock),\allowbreak (\lock, \place', \lock_2)$ in the lock graph.
Therefore, the lock graph also has a cycle $c'$.
%
%
The places $p, p'$ are the places associated with the final steps of the threads
in $E$ that are involved in the cycle $c$ in the LAG.
Hence, by Lemma~\ref{lem:nonconcurrent}, it follows that
$\textsf{all\_concurrent}(c')$.
\end{proof}

\end{theappendix}

\end{document}